\documentclass[10pt]{article}
\usepackage{amsmath}
\usepackage{amssymb}
\usepackage{graphicx}
\usepackage{color}
\usepackage{amsthm}
\usepackage{amsfonts}
\usepackage{amscd}
\usepackage{mathrsfs}
\usepackage{bm}
\usepackage{enumerate}
\usepackage{algorithmic, algorithm}

\newcommand{\D}{\displaystyle}

\newcommand{\bx}{\mathbf{x}}
\newcommand{\bn}{\mathbf{n}}
\newcommand{\by}{\mathbf{y}}

\newcommand{\p}{\partial}
\newcommand{\al}{\alpha}
\newcommand{\e}{\epsilon}

\newcommand{\hkxpj}{R_t(\bx, {\bf x}_j)}

\newcommand{\rhkxpj}{\bar{R}_t(\bx, {\bf x}_j)}

\newcommand{\bfu}{{\bf u}}
\newcommand{\bfp}{{\bf x}}

\newcommand{\bz}{\mathbf{z}}

\newcommand{\M}{{\mathcal M}}

\newcommand{\mathd}{\mathrm{d}}

\newtheorem{theorem}{\textbf{Theorem}}[section]
\newtheorem{lemma}{\textbf{Lemma}}[section]

\newtheorem{proposition}{\textbf{Proposition}}[section]
\newtheorem{assumption}{\textbf{Assumption}}%[section]

\newtheorem{corollary}{\textbf{Corollary}}[section]

\newcommand{\R}{\mathbb{R}}

\numberwithin{equation}{section}

\makeatother

\begin{document}

\title{Convergence of Laplacian spectra from random samples}

\author{
Zuoqiang Shi%
\thanks{Yau Mathematical Sciences Center, Tsinghua University, Beijing, China,
100084. \textit{Email: zqshi@math.tsinghua.edu.cn.}%
}
% \and Jian Sun %
% \thanks{Mathematical Sciences Center, Tsinghua University, Beijing, China,
% 100084. \textit{Email: jsun@math.tsinghua.edu.cn.}%
% }
\thanks{
This research was supported by NSFC Grant 11201257 and 11371220.}
}

 \maketitle

\begin{abstract}
Eigenvectors and  eigenvalues of  discrete  graph  Laplacians are often used
for manifold learning and nonlinear dimensionality reduction. It was previously proved by Belkin and Niyogi \cite{CLEM_08}
 that the eigenvectors and eigenvalues of the graph Laplacian converge to the eigenfunctions
and eigenvalues of the Laplace-Beltrami operator of the manifold in the limit of infinitely many
data points sampled independently from the uniform distribution over the manifold.  Recently, we
introduced Point Integral method (PIM) \cite{LSS,SS-neumann} to solve elliptic equations and corresponding eigenvalue problem
on point clouds. We have established a unified
framework to approximate the elliptic differential operators on point clouds.
In this paper, we prove that the eigenvectors
and eigenvalues obtained by PIM
converge in the limit of infinitely many random samples independently from a distribution (not necessarily to be uniform distribution).
Moreover, one estimate of the rate of the convergence is also given.
\end{abstract}

%\newpage

\section{Introduction}
\label{sec:intro}
The Laplace-Beltrami operator (LBO) is a fundamental object associated to Riemannian manifolds,
which encodes all intrinsic geometry of the manifolds and has many desirable properties. It is also
related to diffusion and heat equation on the manifold, and is connected to a large body of classical
mathematics (see, e.g.,~\cite{rosenberg1997lrm}).
In recent years, the Laplace-Beltrami operator has attracted much attention in many applied fields, including
machine learning, data analysis, computer graphics and computer vision, and geometric modeling and processing. For instance,
the eigensystem of the Laplace-Beltrami operator has been used for representing data in machine learning
and data analysis for dimensionality reduction~\cite{belkin2003led, Coifman05geometricdiffusions}, and
for representing shapes in computer graphics and computer vision for the analysis of images and
3D models~\cite{OvsjanikovSG08, dgw-spec}.

In general, the underlying Riemannian manifold is unknown and often given by a set of sample points.
Thus, in order to exploit the nice properties of the Laplace-Beltrami operator, it is necessary to derive
% a discrete approximation from the given sample points.
% We consider the case where the underlying manifold is a submanifold isometrically embedded in a Euclidean space.
% The set of sample points in a Euclidean space is often referred as a point cloud.
% We emphasize that the submanifold assumption is of great practical importance as in many applications,
% each point is a collection of numerical measurements and different measurements are locally related.
In this paper, we assume that the data points, $X=\{\bx_1,\cdots,\bx_n\}$, are sampled independently over the manifold $\M$ from a
probability distribution $p(\bx)$. On the sample points, we consider following discrete eigenvalue problem.
%Recently, Shi and Sun proposed Point Integral method to approximate elliptic differential operators on point cloud \cite{SSL,SS14}.
\begin{eqnarray}
  \label{eqn:eigen_dis}
  \frac{1}{t}\sum_{j=1}^n R\left(\frac{\|\bfp_i-\bfp_j\|^2}{4t}\right)(u_i-u_j)=
\lambda\sum_{j=1}^n\bar{R}\left(\frac{\|\bfp_i-\bfp_j\|^2}{4t}\right)u_j,
\end{eqnarray}
where $R: \mathbb{R}^+\rightarrow \mathbb{R}^+$ is a kernel function 
% \begin{eqnarray}
%   \label{eq:kernel}
%  R_t(\bx,\by)=\frac{1}{(4\pi t)^{k/2}}R\left(\frac{\|\bx-\by\|^2}{4t}\right),\quad  \bar{R}_t(\bx,\by)=\frac{1}{(4\pi t)^{k/2}}\bar{R}\left(\frac{\|\bx-\by\|^2}{4t}\right)
% \end{eqnarray}
, $\bar{R}(r)=\int_{r}^{+\infty}R(s)\mathd s$.% $t$ is a parameter, which is determined by the desensity of the point cloud in the real computations.

This eigenvalue problem is closely related with the eigenvalue problem of normalized graph Laplacian. 
The graph Laplacian is a discrete object associated to a graph, which reveals many
properties of the graph as does the Laplace-Beltrami operator to the manifold~\cite{Chung}.
% The only difference is on the right hand side. In \eqref{eqn:eigen_dis}
% the weight is distributed over sample points through a kernel function $\bar{R}$. If the weight is concentrated on the diagonal, we just get the normalized graph Laplacian 
% \begin{eqnarray*}
% %  \label{eqn:eigen_dis}
%   \frac{1}{t}\sum_{j=1}^nR\left(\frac{\|\bfp_i-\bfp_j\|^2}{4t}\right)(u_i-u_j)=\lambda u_i \sum_{j=1}^n
% \bar{R}\left(\frac{\|\bfp_i-\bfp_j\|^2}{4t}\right).
% \end{eqnarray*}
 In the presence of no boundary and the sample points are uniformly distributed, Belkin and Niyogi~\cite{CLEM_08}
showed that the spectra of the normalized graph Laplacian converges to
the spectra of $\Delta_\M$. When there is a boundary,  it was
observed in~\cite{Lafon04diffusion, BelkinQWZ12} that
the integral Laplace operator $L_t$ is dominated by the first order derivative and
thus fails to be true Laplacian near the boundary.
Recently, Singer and Wu~\cite{Singer13} showed the spectral
convergence in the presence of the Neumann boundary. In the previous approaches,
the convergence analysis is based on the connection between the graph Laplacian and
the heat operator.
The analysis in this paper is very different from the previous ones.
We consider this problem from the
point of view of solving the Poisson equation on submanifolds, which opens up many tools in the
numerical analysis for studying the graph Laplacian.

The purpose of this paper is to study the behavior of discrete eigenvalue problem \eqref{eqn:eigen_dis} at the limit of $n\rightarrow \infty$ and $t\rightarrow 0$.
The main contribution of this paper is that our study reveals that when $n\rightarrow \infty$ and $t\rightarrow 0$, 
the spectral of \eqref{eqn:eigen_dis} converge to the spectra of following eigenvalue problem.
\begin{eqnarray}
\label{eqn:eigen_neumann}
  \left\{\begin{array}{rcll}
      -\frac{1}{p^2(\bx)}\;\text{div}\left(p^2(\bx)\nabla u(\bx)\right)&=&\lambda u(\bx),&\bx\in \mathcal{M},\\
      \frac{\p u}{\p \bn}(\bx)&=&0,& \bx\in \p \mathcal{M}.
\end{array}\right.
\end{eqnarray}
where $\bn$ is the out normal vector of $\mathcal{M}$.

To analyze the convergence, we introduce an intermediate integral equation.  
\begin{eqnarray}
\frac{1}{t}\int_\mathcal{M} R\left(\frac{\|\bx-\by\|^2}{4t}\right)(u(\bx)-u(\by)) p(\by)\mathd \by=
\int_{\mathcal{M}} \bar{R}\left(\frac{\|\bx-\by\|^2}{4t}\right) f(\by) p(\by)\mathd \by,\quad \bx\in \M.
\label{eq:integral}
\end{eqnarray}
Similar integral equation also can be found in previous works. However,
the rest of the analysis in this paper is very different as the previous ones. 
Before presenting the main results, we need to define three solution operators
%Instead of the operators, $L_{t,n}, L_t$ and $-\frac{1}{p^2(\bx)}\;\text{div}\left(p^2(\bx)\nabla\right)$, 
%we consider the solution operators 
$T, T_t$ and $T_{t,n}$.% which are defined in Section \ref{sec:assume}.

\subsection{Solution operators}
\label{sec:sol-opt}
The solution operators are defined as following.
\begin{itemize}
\item $T: L^2(\mathcal{M})\rightarrow H^2(\mathcal{M})$ is the solution operator of
the problem \eqref{eqn:neumann},
i.e., for any $f\in L^2(\mathcal{M})$, $T(f)$ with $\int_\M T(f) = 0$ is the solution of the following problem:
\begin{eqnarray}
\label{eqn:neumann}
  \left\{\begin{array}{rcll}
      -\frac{1}{p^2(\bx)}\;\text{div}\left(p^2(\bx)\nabla u(\bx)\right)&=&f(\bx),&\bx\in \mathcal{M},\\
      \frac{\p u}{\p \bn}(\bx)&=&0,& \bx\in \p \mathcal{M}.
\end{array}\right.
\end{eqnarray}
where $\bn$ is the out normal vector of $\mathcal{M}$.

\item $T_t:L^2(\mathcal{M})\rightarrow L^2(\mathcal{M})$ is the solution operator of following integral equation \eqref{eqn:integral},
 i.e. $u = T_t(f)$ with $\int_\M u(\bx)p(\bx)\mathd \bx = 0$ solves
the following integral equation
\begin{eqnarray}
\label{eqn:integral}
\frac{1}{t}\int_\mathcal{M} R_t(\bx, \by)(u(\bx)-u(\by)) p(\by)\mathd \by=
\int_{\mathcal{M}} \bar{R}_t(\bx, \by) f(\by) p(\by)\mathd \by.\quad
\end{eqnarray}
where
\begin{align*}
  R_t(\bx, \by)=\frac{1}{(4\pi t)^{k/2}}R\left(\frac{\|\bx-\by\|^2}{4t}\right),\quad 
\bar{R}_t(\bx, \by)=\frac{1}{(4\pi t)^{k/2}}\bar{R}\left(\frac{\|\bx-\by\|^2}{4t}\right).
\end{align*}
\item $T_{t,n}:C(\mathcal{M})\rightarrow C(\mathcal{M})$ is defined as follows.
\begin{eqnarray}
  T_{t,n}(f)(\bx)=\frac{1}{n\,w_{t,n}(\bx)}\sum_{j=1}^nR_t(\bx,\bx_j)u_j+
\frac{t}{n\,w_{t,n}(\bx)}\sum_{j=1}^n\bar{R}_t(\bx,\bx_j)f(\bx_j)
\end{eqnarray}
where $w_{t,n}(\bx)=\frac{1}{n}\sum_{j=1}^nR_t(\bx,\bx_j)$ and $\bfu=(u_1, \cdots, u_n)^t$ with $\sum_{i=1}^n u_i =0$ solves
following linear system,
\begin{eqnarray}
\label{eqn:dis}
 \frac{1}{nt}\sum_{j=1}^nR_t(\bx_i,\bx_j)(u_i-u_j)=\frac{1}{n}\sum_{j=1}^n\bar{R}_t(\bx_i,\bx_j)f(\bx_j)
\end{eqnarray}
\end{itemize}
To simplify the notations, we also introduce two operators. For any $f\in L^2(\M)$,
\begin{align}
  \label{eq:opt-Lt}
  L_tf(\bx)=&\frac{1}{t}\int_\mathcal{M} R_t(\bx, \by)(f(\bx)-f(\by)) p(\by)\mathd \by.
\end{align}
and for any $f\in C(\M)$,
\begin{align}
  \label{eq:opt-Ltn}
  L_{t,n}f(\bx)=&\frac{1}{nt}\sum_{j=1}^n R_t(\bx, \bx_j)(f(\bx)-f(\bx_j)).
\end{align}
Using these definitions, we have that 
\begin{align}
  \label{eq:relation-inverse-Lt}
  L_t(T_tf)(\bx)=\int_{\M}\bar{R}_t(\bx, \by) f(\by) p(\by)\mathd \by
\end{align}
and 
\begin{align}
  \label{eq:relation-inverse-Ltn}
  L_t(T_{t,n}f)(\bx_i)=\frac{1}{n}\sum_{j=1}^n\bar{R}_t(\bx_i,\bx_j)f(\bx_j).
\end{align}
 From \eqref{eq:relation-inverse-Lt} and \eqref{eq:relation-inverse-Ltn}, we can see that 
in some sense, solution operators, $T_t, T_{t,n}$, are inverse operators of 
$L_{t,n}, L_t$. So, it is natural to imagine that their spectra are equivalent.
\begin{proposition}
Let $\bm{\theta}(u)$ denote the
restriction of function $u$ to the sample points $X=(\bx_1,\cdots,\bx_n)^t$, 
i.e., $\bm{\theta}(u) = (u(\bfp_1), \cdots, u(\bfp_n))^t$.
\begin{itemize}
\item[1. ] If a function $u$ is an eigenfunction of $T_{t,n}$
with eigenvalue $\lambda$, then the vector $\bm{\theta}(u)$ is an eigenvector of the eigenproblem~\eqref{eqn:eigen_dis}
with eigenvalue $1/\lambda$.
\item[2. ] If a vector $\mathbf{u}$ is an eigenvector of the eigenproblem~\eqref{eqn:eigen_dis}
with the eigenvalue $\lambda$, then $I_{\lambda}(\mathbf{u})$ is an eigenfunction of $T_{t,n}$ with eigenvalue $1/\lambda$, where
\begin{equation*}
I_{\lambda}(\bfu) (x) = \frac{ \lambda t \sum_{j=1}^n \rhkxpj u_j+\sum_{j=1}^n \hkxpj u_j }
 {\sum_{j=1}^n \hkxpj }.
\end{equation*}
\item[3. ] All eigenvalues of $T, T_t, T_{t,n}$ are real numbers. All generalized eigenvectors of $T, T_t, T_{t,n}$ are eigenvectors.
\end{itemize}
\label{prop:eigen_intergral_dis}
\end{proposition}
This proposition can be proved by following the same line as that in \cite{SS-rate}.

Using this proposition, we only need to analyze the relation among the spectra of $T$ and $T_{t,n}$. In the analysis, the operator $T_t$ plays very important 
role which bridge $T$ and $T_{t,n}$.  
The main advantage of using these solution operators instead of $L_t$ and $L_{t,n}$ is that they are compact operators which is proved in following proposition.
\begin{proposition}
For any $t>0,\; n>0$, $T, T_t$ are compact operators on $H^1(\M)$ into $H^1(\M)$; $T_t, T_{t,n}$ are compact operators on $C^1(\M)$ into $C^1(\M)$.
\end{proposition}
\begin{proof}
   First, it is well known that $T$ is compact operator. $T_{t,h}$ is actually finite dimensional operator, so
it is also compact. To show the compactness of $T_t$, we need the following formula,
\begin{eqnarray}
  T_t u = \frac{1}{w_t(\bx)}\int_\M R_t(\bx,\by)T_tu(\by)\mathd \by+ \frac{t}{w_t(\bx)}\int_\M \bar{R}_t(\bx,\by)u(\by)\mathd \by, \quad \forall u\in H^1(\M).\nonumber
\end{eqnarray}
Using the assumption that $R\in C^2$, direct calculation would gives that
 that $T_tu\in C^2$. This would imply the compactness of $T_t$ both in $H^1$ and $C^1$.
\end{proof}
It is well known that compact operator has many good properties.  Many powerful theorems in the 
spectral theory of compact operators can be used which makes our analysis concise and clear.

\subsection{Main result}
\label{sec:assume}
The main result in this paper is stated with the help of the Riesz spectral projection. 
Let $X$ be a complex Banach space and $L:X\rightarrow X$ be a compact linear operator. The resolvent set $\rho(L)$ is given by the complex numbers $z\in \mathbb{C}$
such that $z-L$ is bijective. The spectrum of $L$ is $\sigma(L)=\mathbb{C}\backslash\rho(L)$. It is well known that $\sigma(L)$ is a countable set with
no limit points other than zero. All non-zero value s in $\sigma(L)$ are eigenvalues.
If $\lambda$ is a nonzero eigenvalue of $L$, the ascent multiplicity $\al$ of $\lambda-L$ is the smallest integer such that $\ker(\lambda-L)^\al=\ker(\lambda-L)^{\al+1}$.

Given a closed smooth curve $\Gamma\subset \rho(L)$ which encloses the eigenvalue $\lambda$ and no other elements of $\sigma(L)$, the Riesz spectral projection associated with
$\lambda$ is defined by
\begin{eqnarray}
  E(\lambda, L)=\frac{1}{2\pi i}\int_{\Gamma} (z-L)^{-1}\mathd z,
\end{eqnarray}
where $i=\sqrt{-1}$ is the unit imaginary. The definition does not depend on the chosen of $\Gamma$.
It is well known that $E(\lambda,L):X\rightarrow X$ has following properties:
\begin{itemize}
\item[1. ] $E(\lambda,L)\circ E(\lambda,L)=E(\lambda,L)$,\quad $L\circ E(\lambda,L)=E(\lambda,L)\circ L$, $E(\lambda,L)\circ E(\mu,L)=0,\;\text{if}\; \lambda\ne \mu$.
\item[2. ] $E(\lambda,L)X=\ker (\lambda-L)^\al$, where $\al$ is the ascent multiplicity of $\lambda-L$.
\item[3. ] If $\Gamma\subset \rho(L)$ encloses more eigenvalues $\lambda_1,\cdots,\lambda_m$, then
$$\D E(\lambda_1,\cdots,\lambda_m,L)X=\oplus_{i=1}^m\ker (\lambda_1-L)^{\al_i}$$
 where $\al_i$ is the ascent multiplicity of $\lambda_i-L$.
\end{itemize}
The properties (2) and (3) are of fundamental importance for the study of eigenvector approximation.

To prove the convergence, we need some assumptions on the manifold $\M$, probability distribution $p(\bx)$ and the kernel function $R$ which are summarized as following:
\begin{assumption}
\label{assumptions}
\begin{itemize}
\item[]
\item Assumptions on the manifold: $\M$ is $k$-dimensional compact and $C^\infty$ smooth manifold isometrically embedded in a Euclidean space $\mathbb{R}^d$.
\item Assumptions on the sample points: $X=\{\bfp_1,\cdots,\bfp_n\}$ are sampled independently over the manifold $\M$ distribution $p(\bx)\in C^1(\M)$
and $\min_{\bx\in \M}p(\bx)>0$, $\max_{\bx\in \M}p(\bx)<\infty$.
\item Assumptions on the kernel function $R(r)$:
\begin{itemize}
\item[(a)] $R\in C^2(\mathbb{R}^+)$;
\item[(b)]
%$R(r) \le 1$ for $\forall r \in \R^+$ and $R(r) = 0$ for $\forall r >1$.
$R(r)\ge 0$ and $R(r) = 0$ for $\forall r >1$;
\item[(c)]
 $\exists \delta_0>0$ so that $R(r)\ge\delta_0$ for $0\le r\le\frac{1}{2}$.
\end{itemize}
%\item Assumptions on $t$ and $h$: $t$ and $h/\sqrt{t}$ are small enough, i.e., $t$ and $h/\sqrt{t}$ are less than a positive constant
%which only depends on $\M$.
\end{itemize}
\end{assumption}

Now, we are ready to state the main theorem. Since $T$ and $T_{t,n}$ are both compact operators, their eigenvalues can be sorted as
 \begin{eqnarray}
  &&0<\cdots\le\lambda_i\le \cdots\le \lambda_2\le \lambda_1,\nonumber\\
  &&0<\cdots\le\lambda_i^{t,n}\le \cdots\le \lambda_2^{t,n}\le \lambda_1^{t,n},\nonumber
\end{eqnarray}
where the same eigenvalue is repeated according to its multiplicity.

For corresponding eigenvalues and eigenfunctions, we have following theorem.
% The following theorem guarantees the convergence of our finite integral method for
% solving the eigensystem problem~\eqref{eqn:eigen_neumann} of the Laplace-Beltrami operator
% with Neumann boundary condition .
\begin{theorem}
Under the assumptions in Assumption \ref{assumptions},
let $\lambda_i$ be the $i$th largest eigenvalue of $T$ (same eigenvalue is repeated according to its multiplicity) 
with multiplicity $\alpha_i$ and $\phi_i^k, k=1,\cdots,\alpha_i$
be the linear independent eigenfunctions corresponding to $\lambda_i$. Let $\lambda_{i}^{t,n}$ be the $i$th largest eigenvalue of
$T_{t,n}$. With probability at least $1-1/n$, there exists a constant $C_1>0, C_2>0$ depend on $\M$, kernel function $R$,
distribution $p$ and spectra of $T$, such that
$$|\lambda_i^{t,n} - \lambda_i|\le C_1 \left(t^{1/2}+\frac{\log n+|\log t|+1}{t^{k+3}\sqrt{n}}\right) , $$
and %there exist another constant $C$ such that, for any $\phi\in E(\lambda_i,T)X$ and $X=H^1(\M)$,
$$\|\phi_i^k-E(\sigma_{i}^{t,n},T_{t,n})\phi_i^k\|_{H^1(\M)} \le C_2 \left(t^{1/2}+\frac{\log n+|\log t|+1}{t^{k+2}\sqrt{n}}\right) ,$$
as long as $n$ large enough.
Here $\sigma_i^{t,n}=\{\lambda_j^{t,n}\in \sigma(T_{t,n}): j\in I_i\}$ and $I_i=\{j\in \mathbb{N}: \lambda_j=\lambda_i\}$.
\label{thm:eigen_neumann}
\end{theorem}

This theorem will be proved in Section 2 and 3. Some conclusions are made in Section 4.

% \section{Notations}
% \label{sec:assume}
% Our main theoretical contribution is to establish the convergence results for our finite integral method
% for solving the problem~\eqref{eqn:neumann} and the problem~\eqref{eqn:eigen_neumann}.
% To state the convergence result, we need to describe the assumptions on the integral kernel $R$, the regularity of the submanifold $\M$ and
% how well the input data approximate the submanifold $\M$.

\section{Proof of the main theorem (Theorem \ref{thm:eigen_neumann})}
\label{sec:intermediate}

The proof of Theorem \ref{thm:eigen_neumann} mainly consists of three parts. 
The first part is to relate the difference of the eigenvalues and eigenfunctions with the
difference of operators $T-T_t$ and $T_t-T_{t,n}$ (Theorem \ref{thm:eigen_operator}). This is 
achieved by using one theorem in the perturbation of compact operators.

To apply the theorem obtained in the first part, we need to estimate the difference of operators $T-T_t$ and $T_t-T_{t,n}$
in $H^1$ and $C^1$ norm respectively. This is also the most difficult part. Comparing with the pointwise convergence
which was proved in previous works, convergence in norm is much stronger and  much more difficult to prove.
Fortunately, under some mild assumption which are listed in Assumption \ref{assumptions}, we could prove that 
$T_t\rightarrow T$ in $H^1$ norm as $t\rightarrow 0$ (Theorem \ref{thm:converge_h1}) and $T_{t,n}\rightarrow T_t$ in $C^1$ norm as $n\rightarrow \infty$ 
(Theorem \ref{thm:converge_c1}). 

To get the rate of the convergence, in the last part of the analysis,
we use the the theory of the Glivenko-Cantelli class in statistical learning to estimate the error in the Mote-Carlo 
integration. 
The key point in this part is to estimate the covering number of the function classes defined as following.

Here, we list some notations which will be used in the proof. Some of them have been defined in previous sections. We also list them here for the convenience of readers.
\begin{itemize}
\item $k$: dimension of the underlying manifold; $d$: dimension of the ambient Euclidean space;
\item $C$: positive constant independent on $t$ and sample points $X_n$. We abuse the notation to denote all the constants independent on $t$ and sample points $X_n$
by $C$. It may be different in different places. %It is determined by the manifold $\M$, kernel function $R$ and distribution $p$
\item $C_t=\frac{1}{(4\pi t)^{k/2}}$ is the normalize constant of kernel function $R$.
\item $p(\bx)$: probability distribution function.
\item $R$: kernel function. $\bar{R}(r)=\int_{r}^\infty R(s)\mathd s$.
\item $R_t(\bx, \by)=\frac{1}{(4\pi t)^{k/2}}R\left(\frac{\|\bx-\by\|^2}{4t}\right),\quad 
\bar{R}_t(\bx, \by)=\frac{1}{(4\pi t)^{k/2}}\bar{R}\left(\frac{\|\bx-\by\|^2}{4t}\right).$
\item $w_t(\bx)=\int_\M R_t(\bx,\by)\mathd \by$, $\quad w_{t,n}(\bx)=\frac{1}{n(4\pi t)^{k/2}}\sum_{j=1}^n R\left(\frac{|\bx-\bx_j|^2}{4t}\right)$
\item $w_{\min}, w_{\max}$: $\D w_{\min}=\inf_{t>0}\min_{\bx\in \M}w_t(\bx)$, $\D w_{\max}=\sup_{t>0}\max_{\bx\in \M}w_t(\bx)$. 
Under the assumption in Assumption \ref{assumptions}, 
we can show that $0< w_{\min}, w_{\max}<\infty$.
\item $p(f)=\int_\M f(\bx)p(\bx)\mathd \bx,\quad p_n(f)=\frac{1}{n}\sum_{i=1}^nf(\bx_i).$
\item $\mathcal{R}_t=\left\{R\left(\frac{|\bx-\by|^2}{4t}\right):\bx\in \M\right\}$
\item $\overline{\mathcal{R}}_t=\left\{\bar{R}\left(\frac{|\bx-\by|^2}{4t}\right):\bx\in \M\right\}$
\item $\mathcal{D}_t=\left\{\nabla_\bx R\left(\frac{|\bx-\by|^2}{4t}\right):\bx\in \M\right\}$
%\item $\overline{\mathcal{D}}_t=\left\{\nabla_\bx \bar{R}\left(\frac{|\bx-\by|^2}{4t}\right):\bx\in \M\right\}$
\item $\mathcal{R}_t\cdot \mathcal{K}_{t,n}=\left\{\frac{1}{w_{t,n}(\by)}R\left(\frac{|\bx-\by|^2}{4t}\right)
R\left(\frac{|\bz-\by|^2}{4t}\right):\bx\in \M,\; \bz\in \M\right\}$
\item $\overline{\mathcal{R}}_t\cdot \mathcal{K}_{t,n}=\left\{\frac{1}{w_{t,n}(\by)}R\left(\frac{|\bx-\by|^2}{4t}\right)
\bar{R}\left(\frac{|\bz-\by|^2}{4t}\right):\bx\in \M,\; \bz\in \M\right\}$
\item $\overline{\mathcal{R}}_t\cdot \overline{\mathcal{K}}_{t,n}=\left\{\frac{1}{w_{t,n}(\by)}\bar{R}\left(\frac{|\bx-\by|^2}{4t}\right)
\bar{R}\left(\frac{|\bz-\by|^2}{4t}\right):\bx\in \M,\; \bz\in \M\right\}$
\item $\mathcal{D}_t\cdot \mathcal{K}_{t,n}=\left\{\frac{\sqrt{t}}{w_{t,n}(\by)}R\left(\frac{|\bx-\by|^2}{4t}\right)
\nabla_\bz R\left(\frac{|\bz-\by|^2}{4t}\right):\bx\in \M,\; \bz\in \M\right\}$
\item $\overline{\mathcal{D}}_t\cdot \mathcal{K}_{t,n}=\left\{\frac{\sqrt{t}}{w_{t,n}(\by)}R\left(\frac{|\bx-\by|^2}{4t}\right)
\nabla_\bz\bar{R}\left(\frac{|\bz-\by|^2}{4t}\right):\bx\in \M,\; \bz\in \M\right\}$
\item $\overline{\mathcal{D}}_t\cdot \overline{\mathcal{K}}_{t,n}=\left\{\frac{\sqrt{t}}{w_{t,n}(\by)}\bar{R}\left(\frac{|\bx-\by|^2}{4t}\right)
\nabla_\bz\bar{R}\left(\frac{|\bz-\by|^2}{4t}\right):\bx\in \M,\; \bz\in \M\right\}$
\end{itemize}
% \begin{eqnarray}
%   \mathcal{R}_t&=&\left\{R\left(\frac{|\bx-\by|^2}{4t}\right):\bx\in \M\right\},\quad 
%   \overline{\mathcal{R}}_t=\left\{\bar{R}\left(\frac{|\bx-\by|^2}{4t}\right):\bx\in \M\right\},\nonumber\\
%   \mathcal{K}_{t,n}&=&\left\{\frac{1}{w_{t,n}(\by)}R\left(\frac{|\bx-\by|^2}{4t}\right):\bx\in \M\right\},\quad
% %\bar{\mathcal{K}}_{t,n}=\left\{\frac{1}{w_{t,n}(\by)}\bar{R}\left(\frac{|\bx-\by|^2}{4t}\right):\bx\in \M\right\}\nonumber\\
% \mathcal{D}_t=\left\{\nabla_\bx R\left(\frac{|\bx-\by|^2}{4t}\right):\bx\in \M\right\},\nonumber\\
%   \mathcal{R}_t\cdot \overline{\mathcal{R}}_t&=&
% \left\{R\left(\frac{|\bx-\by|^2}{4t}\right)\bar{R}\left(\frac{|\bz-\by|^2}{4t}\right):\bx\in \M,\; \bz\in \M\right\},\nonumber\\
% \mathcal{R}_t\cdot \mathcal{R}_t&=&
% \left\{R\left(\frac{|\bx-\by|^2}{4t}\right)R\left(\frac{|\bz-\by|^2}{4t}\right):\bx\in \M,\; \bz\in \M\right\},\nonumber\\
% \mathcal{R}_t\cdot \mathcal{K}_{t,n}&=&\left\{\frac{1}{w_{t,n}(\by)}R\left(\frac{|\bx-\by|^2}{4t}\right)
% R\left(\frac{|\bz-\by|^2}{4t}\right):\bx\in \M,\; \bz\in \M\right\}\nonumber
% \end{eqnarray}
%\begin{eqnarray}
%   K_t(\bx,\by)=R_t(\bx,\by)/w_t(\by),\quad \bar{K}_t(\bx,\by)=\bar{R}_t(\bx,\by)/w_t(\by),
% \end{eqnarray}

\subsection{Perturbation results of solution operators}

% \begin{theorem}
% \label{lem:resovent}
% Let $(X,\|\cdot\|_X)$ be an arbitrary Banach space. Let $S$ and 
% $T$ be compact linear operators on $X$ into $X$.
%  Let $z\in \rho(T)$. Assume 
% \begin{eqnarray}
%   \|T-S\|_{X} \le \frac{1}{2\|(z-T)^{-1}\|_{X}}. 
% \end{eqnarray}
% Then $z\in \rho(S)$ and $(z-S)^{-1}$ has the bound
% \begin{eqnarray}
%   \|(z-S)^{-1}\|_{X}\le 2\|(z-T)^{-1}\|_{X}.
% \end{eqnarray} 
% \end{theorem}

First, we need some results regarding the perturbation of the compact operators.
\begin{theorem} (\cite{Atkinson67})
\label{lem:resovent}
Let $(X,\|\cdot\|_X)$ be an arbitrary Banach space. Let $S$ and 
$T$ be compact linear operators on $X$ into $X$.
 Let $z\in \rho(T)$. Assume 
\begin{eqnarray}
  \|(T-S)S\|_{X} \le \frac{|z|}{\|(z-T)^{-1}\|_{X}}. 
\end{eqnarray}
Then $z\in \rho(S)$ and $(z-S)^{-1}$ has the bound
\begin{eqnarray}
  \|(z-S)^{-1}\|_{X}\le \frac{1+\|S\|_X\|(z-T)^{-1}\|_{X}}{|z|-\|(z-T)^{-1}\|_{X}\|(T-S)S\|_{X}}.
\end{eqnarray} 
\end{theorem}

\begin{theorem}  (\cite{Atkinson67})
\label{thm:converge_evector}
  Let $(X,\|\cdot\|_X)$ be an arbitrary Banach space. Let $S$ and 
$T$ be compact linear operators on $X$ into $X$. Let $z_0\in \mathbb{C}, \; z_0\ne 0$ and let $\epsilon>0$ be less than $|z_0|$, 
denote the circumference $|z - z_0|=\epsilon$ by $\Gamma$ and assume $\Gamma\subset \rho(T)$. 
Denote the interior of $\Gamma$ by $U$. Let $\sigma_T=U\cap \sigma(T)\ne \emptyset$. $\sigma_S=U\cap\sigma(S)$. 
 Let $E(\sigma_S,S)$ and $E(\sigma_T,T)$ be the corresponding spectral projections of $S$ for $\sigma_S$ and $T$ for $\sigma_T$, i.e.
 \begin{eqnarray}
   E(\sigma_S,S)=\frac{1}{2\pi i}\int_{\Gamma}(z-S)^{-1}\mathd z,\quad    E(\sigma_T,T)=\frac{1}{2\pi i}\int_{\Gamma}(z-T)^{-1}\mathd z.
 \end{eqnarray}
Assume 
\begin{eqnarray}
 \|(T-S)S\|_X\le \min_{z\in \Gamma} \frac{|z|}{\|(z-T)^{-1}\|_X}
\end{eqnarray}
Then, we have
\begin{itemize}
\item[(1).] Dimension $E(\sigma_S,S)X = E(\sigma_T,T)X$, thereby $\sigma_S$ is nonempty and of the same multiplicity as $\sigma_T$.
\item[(2).] For every $x\in  X$,
\begin{eqnarray}
  \|E(\sigma_T,T)x-E(\sigma_S,S)x\|_X \le \frac{M\epsilon}{c_0}\left(\|(T-S)x\|_X+\|x\|_X\|(T-S)S\|_X\right).\nonumber
\end{eqnarray}
where $M=\max_{z\in \Gamma}\|(z-T)^{-1}\|_X$, $c_0=\min_{z\in \Gamma}|z|$.
\end{itemize}

\label{thm:eigen_perturb}
\end{theorem}

\begin{lemma}(\cite{SS-rate})
\label{lem:norm_inverse_h1}
 Let $T$ be the solution operator of
the Neumann problem \eqref{eqn:neumann} and $z\in \rho(T)$, then
  \begin{eqnarray}
    \|(z-T)^{-1}\|_{H^1(\M)}\le \max_{n\in \mathbb{N}} \frac{1}{|z-\lambda_n|},\nonumber
  \end{eqnarray}
where $\{\lambda_n\}_{n\in \mathbb{N}}$ is the set of eigenvalues of $T$.
% $S$ a compact linear operator
% on $H^1(\M)$ into $H^1(\M)$.
%   \begin{itemize}
%   \item[(1).] $T$ and $\{T_t\}_{t>0}$ are all compact operators on $H^1(\M)$ into $H^1(\M)$.
% \item[(2).]
%   \begin{eqnarray}
%     \|(\lambda-T)^{-1}\|_{H^1(\M)}\le \max_{n\in \mathbb{N}} \frac{1}{|\lambda-\lambda_n|},
%   \end{eqnarray}
% where $\{\lambda_n\}_{n\in \mathbb{N}}$ is the set of eigenvalues of $T$.
%   \end{itemize}
\end{lemma}

\begin{lemma}(\cite{SS-rate})
\label{lem:norm_inverse_c1}
Let $T_t$ be the solution operator of the integral equation \eqref{eqn:integral}. For any $z\in \mathbb{C}\backslash \bigcup_{n\in \mathbb{N}}B(\lambda_n, r_0)$ 
with $r_0>\|T-T_t\|_{H^1}$, then
  \begin{eqnarray}
   \|(z-T_t)^{-1}\|_{C^1}\le \max\left\{\frac{2|\M|}{ |z|t^{(k+2)/4}} \left(\min_{n\in \mathbb{N}}|z-\lambda_n| - \|T-T_t\|_{H^1}\right)^{-1},\frac{2}{|z|}\right\}.\nonumber
\end{eqnarray}
\end{lemma}

\begin{theorem} (\cite{SS-rate})
Let $T_t$ be the solution operator of the integral equation \eqref{eqn:integral} and $\lambda_n$ be
eigenvalues of $T$, then
\label{thm:converge_ev_Tt}
  \begin{eqnarray}
    \sigma(T_t)\subset \bigcup_{n\in \mathbb{N}}B\left(\lambda_n,2\|T-T_t\|_{H^1(\M)}\right).\nonumber
  \end{eqnarray}
\end{theorem}

%Now we can show the convergence of eigenproblem, Theorem~\ref{thm:eigen_neumann}. 
The main result in this subsection is stated as following in which the difference of the eigenvalues and eigenfunctions are related with the 
difference of the solutions operators.
\begin{theorem}
Let $\lambda_m$ be the $m$th largest eigenvalue of $T$ with multiplicity $\alpha_m$ and $\phi_m^k, k=1,\cdots,\alpha_m$ 
be the eigenfunctions corresponding to $\lambda_m$. %and $h_i=\max_{1\le k\le \alpha_i}h(\phi_i^k)$. 
Let $\lambda_{m}^{t,n}$ be the $m$th largest eigenvalue of 
$T_{t,n}$. Let $\D\gamma_m=\min_{j\le m}|\lambda_j-\lambda_{j+1}|$ and 
\begin{align*}
&  \left\|(T_{t,n}-T_t)T_{t,n}\right\|_{C^1}\le \min\{\frac{t}{2}, \frac{\gamma_mt^{k/4+3/2}}{24}, 
\frac{(|\lambda_m|-\gamma_m/3)^2t^{(k+2)/4}\gamma_m}{12}, \frac{(|\lambda_m|-\gamma_m/3)^2}{2}\},\\
& \|T-T_t\|_{H^1(\M)}\le \gamma_m/12,\quad \|(T-T_t)T_t\|_{H^1(\M)}\le (|\lambda_m|-\gamma_m/3)\gamma_m/3
\end{align*}
Then there exists a constant $C_1, C_2$ depend on $\M$, the kernel function $R$, $\gamma_m$ and $\lambda_m$, such that
$$|\lambda_m^{t,n} - \lambda_m|\le \frac{2}{t^{k/4+3/2}}\left\|(T_{t,n}-T_t)T_{t,n}\right\|_{C^1}+\|T-T_t\|_{H^1(\M)}$$
%C_1 \left(t^{1/2}+\frac{h}{t^{k/2+3}}\right) , $$ 
and %there exist another constant $C$ such that, for any $\phi\in E(\lambda_i,T)X$ and $X=H^1(\M)$, 
$$\|\phi_m^k-E(\sigma_{m}^{t,n},T_{t,n})\phi_m^k\|_{H^1(\M)} \le C(\|(T-T_t)\phi_m^k\|_{H^1}+\|(T-T_t)T_t\|_{H^1}) 
+\frac{C}{t^{(k+2)/4}}(\|(T_t-T_{t,n})\phi_m^k\|_{C^1}+\|(T_t-T_{t,n})T_{t,n}\|_{C^1})$$ %C_2 \left(t^{1/2}+\frac{h_0}{t^{k/2+2}}+\frac{h(\phi_m^k)}{t^{k/2+2}}\right).$$
Here $\sigma_m^{t,n}=\{\lambda_j^{t,n}\in \sigma(T_{t,n}): j\in I_m\}$ and $I_m=\{j\in \mathbb{N}: \lambda_j=\lambda_m\}$.
\label{thm:eigen_operator}
\end{theorem}
\begin{proof} 
 Let $r_1=\frac{2}{t^{k/4+3/2}}\left\|(T_{t,n}-T_t)T_{t,n}\right\|_{C^1}+\|T-T_t\|_{H^1(\M)}$, $\mathcal{A}=\mathbb{C}\backslash \left[\bigcup_{n\in \mathbb{N}}B(\lambda_n, r_1)
\bigcup B(0,t^{1/2})\right]$, 
For any $z\in \mathcal{A}$, using Lemma \ref{lem:resovent}, we have
\begin{eqnarray*}
  \|(z-T_{t})^{-1}\|_{C^1}&\le& \frac{2|\M|}{|z| t^{(k+2)/4}} \left(\min_{n\in \mathbb{N}}|z-\lambda_n| - \|T-T_t\|_{H^1}\right)^{-1}\nonumber\\
&\le & \frac{2|\M|}{ t^{k/4+1}} \left(r_1 - \|T-T_t\|_{H^1}\right)^{-1}\nonumber\\
&=& \frac{t^{1/2}}{\left\|(T_{t,n}-T_t)T_{t,n}\right\|_{C^1}}\le \frac{|z|}{\left\|(T_{t,n}-T_t)T_{t,n}\right\|_{C^1}}
\end{eqnarray*}
or 
\begin{eqnarray*}
  \|(z-T_{t})^{-1}\|_{C^1}&\le& \frac{2}{|z|}\le \frac{2}{t^{1/2}}\le \frac{\sqrt{t}}{\left\|(T_{t,n}-T_t)T_{t,n}\right\|_{C^1}}\le \frac{|z|}{\left\|(T_{t,n}-T_t)T_{t,n}\right\|_{C^1}}.
\end{eqnarray*}
Here, we use the condition that $\left\|(T_{t,n}-T_t)T_{t,n}\right\|_{C^1}\le t/2$.% and Theorem \ref{thm:converge_c1}.

Both above two inequalies implies that
\begin{eqnarray*}
\left\|(T_{t,n}-T_t)T_{t,n}\right\|_{C^1}\le \frac{|z|}{\|(z-T_t)^{-1}\|_{C^1}}.
\end{eqnarray*}
Then using Lemma \ref{lem:resovent}, we have $z\in \rho(T_{t,n})$.

Since $z$ is arbitrary in $\mathcal{A}$, we get $\mathcal{A}\subset \rho(T_{t,n})$. This means that
\begin{eqnarray}
\label{eqn:est_eigen}
  \sigma(T_{t,n})=\mathbb{C}\backslash\rho(T_{t,n})\subset \mathbb{C}\backslash \mathcal{A}
= \bigcup_{n\in \mathbb{N}}B(\lambda_n, r_1)\bigcup B(0,t^{1/2}).
\end{eqnarray}
Moreover, using Theorem \ref{thm:converge_ev_Tt} and the definition of $r_1$, we have
 \begin{eqnarray}
\label{eqn:est_eigen_T}
  \sigma(T_{t})\subset \bigcup_{n\in \mathbb{N}}B(\lambda_n, 2r_1).
\end{eqnarray}
% On the other hand, using Theorem \ref{thm:converge_h1} and \ref{thm:converge_c1}, we know that there exist $C>0$ independent on $t$ and samples $P$, such that 
% \begin{eqnarray}
% \label{eqn:est_radius}
%   r_1\le C\left(t^{1/2}+\frac{h_0}{t^{k/2+3}}\right).
% \end{eqnarray}
For any fixed eigenvalue $\lambda_m\in \sigma(T)$, let 
$\D\gamma_m=\min_{j\le m}|\lambda_j-\lambda_{j+1}|$. Using the structure of $\sigma(T)$, we know that $\gamma_m>0$. 
Since 
\begin{align*}
  \frac{2}{t^{k/4+3/2}}\left\|(T_{t,n}-T_t)T_{t,n}\right\|_{C^1}\le \gamma_m/12,\quad \|T-T_t\|_{H^1(\M)}\le \gamma_m/12,
\end{align*}
we know that 
$r_1<\gamma_m/6$.
% Without loss of generality, 
% we assume $t$ is small enough such that $C\sqrt{t}\le \gamma/12$. And further use the condition that $\frac{Ch_0}{t^{k/2+3}}\le \gamma/12$, we know that 
% $r_1<\gamma/6$.

Let $\Gamma_j=\{z\in \mathbb{C}:|z-\lambda_j|=\gamma_j/3\}$, $U_j$ be the aera enclosed by $\Gamma_j$.
 Let 
$$\sigma_{t,j}=\sigma(T_t)\bigcap U_j,\quad\sigma_{t,n,j}=\sigma(T_{t,n})\bigcap U_j.$$
 
Using the definition of $\Gamma_j$, we know for any $j\le m$,
$\Gamma_j\subset \rho(T), \rho(T_t)$ and $ \rho(T_{t,n})$. 

In order to apply Theorem \ref{thm:converge_evector}, we need to verify the conditions
\begin{eqnarray}
\label{eqn:cond_T}
  \|(T-T_t)T_t\|_{H^1}&\le& \min_{z\in \Gamma_j}\frac{|z|}{\|(z-T)^{-1}\|_{H^1}},\\
\label{eqn:cond_Tt}
\|(T_t-T_{t,n})T_{t,n}\|_{C^1}&\le& \min_{z\in \Gamma_j}\frac{|z|}{\|(z-T_t)^{-1}\|_{C^1}}.
\end{eqnarray}
Using Lemma \ref{lem:norm_inverse_h1} and the choice of $\Gamma_j$, we have
\begin{eqnarray}
  \min_{z\in \Gamma_m}\frac{|z|}{\|(z-T)^{-1}\|_{H^1}}\ge \frac{\min_{z\in \Gamma_m}|z|}{\max_{z\in \Gamma_m}\|(z-T)^{-1}\|_{H^1}}
\ge (|\lambda_m|-\gamma_m/3)\min_{z\in \Gamma_m,n\in \mathbb{N}}|z-\lambda_m|=(|\lambda_m|-\gamma_m/3)\gamma_m/3.\nonumber
\end{eqnarray}
Then, using the assumption that $\|(T-T_t)T_t\|_{H^1(\M)}\le (|\lambda_m|-\gamma_m/3)\gamma_m/3$, 
condition \eqref{eqn:cond_T} is true..

Using Lemma \ref{lem:norm_inverse_c1}, we have
\begin{eqnarray}
  \min_{z\in \Gamma_m}\frac{|z|}{\|(z-T_t)^{-1}\|_{C^1}}&\ge& \frac{\min_{z\in \Gamma_m}|z|}{\max_{z\in \Gamma_m}\|(z-T_t)^{-1}\|_{C^1}}
\nonumber\\
&\ge& \frac{(|\lambda_m|-\gamma_m/3)^2t^{(k+2)/4}}{2}
\left(\min_{z\in \Gamma_m,n\in \mathbb{N}}|z-\lambda_m|-\|T-T_t\|_{H^1}\right)\nonumber\\
&\ge&\frac{(|\lambda_m|-\gamma_m/3)^2t^{(k+2)/4}\gamma_m}{12}.
\label{eq:1}
\end{eqnarray}
or
\begin{eqnarray}
  \min_{z\in \Gamma_m}\frac{|z|}{\|(z-T_t)^{-1}\|_{C^1}}&\ge& \frac{\min_{z\in \Gamma_m}|z|}{\max_{z\in \Gamma_m}\|(z-T_t)^{-1}\|_{C^1}}
\ge \frac{(|\lambda_m|-\gamma_m/3)^2}{2}.
\end{eqnarray}
To get the last inequality of \eqref{eq:1}, we use the assumption that $\|T-T_t\|_{H^1}\le \gamma/6$ and 
$\D\min_{z\in \Gamma_m,n\in \mathbb{N}}|z-\lambda_m|=\gamma_m/3$.

Using the assumption that $\|(T-T_{t,n})T_{t,n}\|_{C^1(\M)}\le \min\{\frac{(|\lambda_m|-\gamma_m/3)^2t^{(k+2)/4}\gamma_m}{12}, \frac{(|\lambda_m|-\gamma_m/3)^2}{2}\} $,
 condition \eqref{eqn:cond_Tt} is satisfied.

Then using Theorem \ref{thm:converge_evector}, we have
\begin{eqnarray}
  \dim(E(\lambda_m,T))=  \dim(E(\sigma_{t,m},T_t)) =  \dim(E(\sigma_{t,n,m},T_{t,n})).
\end{eqnarray}
Using \eqref{eqn:est_eigen}, above equality would imply that
\begin{eqnarray}
  |\lambda_m^{t,n}-\lambda_m|\le r_1=\frac{2}{t^{k/4+3/2}}\left\|(T_{t,n}-T_t)T_{t,n}\right\|_{C^1}+\|T-T_t\|_{H^1(\M)}.
\label{eq:evalue}
%\le C\left(t^{1/2}+\frac{h_0}{t^{k/2+3}}\right).
\end{eqnarray}

The convergence of eigenspace is also given by Theorem \ref{thm:converge_evector}. 
For any 
$x\in E(\lambda_m,T)$, $\|x\|_{C^1}=1$, 
\begin{equation}
  \|x-E(\sigma_{t,m},T_t)x\|_{H^1}\le \frac{\max_{z\in \Gamma_m}\|(z-T)^{-1}\|_{H^1}\gamma_m/3}{\min_{z\in \Gamma_m}|z|}(\|(T-T_t)x\|_{H^1}+\|(T-T_t)T_t\|_{H^1}\|x\|_{H^1}).\nonumber
%\\
 % \|E(\sigma_{t,n},T_t)x-E(\sigma_{t,n,j},T_{t,n})x\|_{C^1}&\le& C\left(\|T_t-T_{t,n}\|_{C^1}+\|(T_t-T_{t,n})T_{t,n}\|_{C^1}\right) \le \frac{Ch}{t^{k/4+2}}.
\end{equation}
Using Lemma \ref{lem:norm_inverse_h1}, we know that 
\begin{align*}
  \max_{z\in \Gamma_m}\|(z-T)^{-1}\|_{H^1}\le \max_{j\in \mathbb{N}}\frac{1}{|z-\lambda_j|}\le \frac{3}{2\gamma_m},
\end{align*}
and $\min_{z\in \Gamma_m}|z|=|\lambda_m|-\gamma_m/3$.
This implies that from Theorems \ref{thm:converge_h1},
\begin{eqnarray}
  \|x-E(\sigma_{t,m},T_t)x\|_{H^1}\le C(\|(T-T_t)x\|_{H^1}+\|(T-T_t)T_t\|_{H^1}\|x\|_{H^1}).
\label{eq:evector-t}
\end{eqnarray}
Regarding the convergence from $T_{t,n}$ to $T_t$, using Theorem \ref{thm:converge_evector} again, we have
\begin{equation}
\D \|E(\sigma_{t,m},T_t)x-E(\sigma_{t,n,m},T_{t,n})x\|_{C^1}\le \frac{\D \gamma_m\max_{z\in \Gamma_m}\|(z-T_t)^{-1}\|_{C^1}}
{3\min_{z\in \Gamma_m}|z|}
\left(\|(T_t-T_{t,n})x\|_{C^1}+\|(T_t-T_{t,n})T_{t,n}\|_{C^1}\right).
\label{eq:evector-n}
\end{equation}
Using Lemma \ref{lem:norm_inverse_c1}, we know that 
\begin{align}
  \max_{z\in \Gamma_m}\|(z-T_t)^{-1}\|_{C^1}\le &\max_{z\in \Gamma_m}
\left\{\frac{2}{ |z|t^{(k+2)/4}} \left(\min_{j\in \mathbb{N}}|z-\lambda_j| - \|T-T_t\|_{H^1}\right)^{-1},\frac{2}{|z|}\right\}
\nonumber\\
\le& \max\left\{\frac{12}{\gamma_m (|\lambda_m|-\gamma_m/3)t^{(k+2)/4}}, \frac{2}{|\lambda_m|-\gamma_m/3}\right\}.
\label{eq:evector-c}
\end{align}
To get the last inequality, we use that $\|T-T_t\|_{H^1}\le \gamma_m/6$ and $|z-\lambda_m|=\gamma_m/3$, $|z|\ge |\lambda_m-\gamma/3|$ for $z\in \Gamma_m$.

Then the proof is completed by combining \eqref{eq:evalue}, \eqref{eq:evector-t}, \eqref{eq:evector-n} 
and \eqref{eq:evector-c}.
% \begin{eqnarray}
%  \|E(\sigma_{t,m},T_t)x-E(\sigma_{t,n,m},T_{t,n})x\|_{C^1}&\le&C\left(\frac{h_0}{t^{k/2+2}}+\frac{h(x)}{t^{k/2+2}}\right) .
% \end{eqnarray}

% Finally, we have
% \begin{eqnarray}
%   \|x-E(\sigma_{t,n,m},T_{t,n})x\|_{H^1}\le C \left(t^{1/2}+\frac{h_0}{t^{k/2+2}}+\frac{h(x)}{t^{k/2+2}}\right).
% \end{eqnarray}

 \end{proof}

\subsection{Convergence of solution operators}

To apply Theorem \ref{thm:eigen_operator}, we need to estimate the difference of the solution operators. More precisely, 
we need to estimate $\|T-T_t\|_{H^1}$ and $\|T_t-T_{t,n}\|_{C^1}$ as $t\rightarrow 0$ and $n\rightarrow \infty$. These results are summarized in 
Theorem \ref{thm:converge_h1} and Theorem \ref{thm:converge_c1} respectively.

%Following theorem shows that $T_t$ converges to $T_t$ in $H^1$ norm as $t\rightarrow 0$ and $T_t$ is uniformly bounded in $H^1$ space.
\begin{theorem} (\cite{SS-iso})
\label{thm:converge_h1}
Under the assumptions in Assumption \ref{assumptions}, there exists a constant $C>0$ only depends on $\M$ and the kernel function $R$, such that
  \begin{eqnarray}
    \|T-T_t\|_{H^1}\le Ct^{1/2},\quad \|T_t\|_{H^1}\le C.\nonumber
  \end{eqnarray}
\end{theorem}
The proof of this theorem can be found in \cite{SS-iso}.

% \begin{theorem}
% \label{thm:bound}
% Under the assumptions in Assumption \ref{assumptions}, there exists a constant $C>0$ only depends on $\M$ and the kernel function $R$, such that
%   \begin{eqnarray}
    
%   \end{eqnarray}
% \end{theorem}

The other theorem is about $\|T_t-T_{t,n}\|_{C^1}$.
\begin{theorem}
\label{thm:converge_c1}
Under the assumptions in Assumption \ref{assumptions} and
\begin{align}
\label{eq:assume-1}
C_t \sup_{f\in \mathcal{R}_{t'}\cup\mathcal{R}_{t}\cup\mathcal{R}_{8t}}|p(f)-p_n(f)|&\le w_{\min}/2,\\
C_t\sup_{f\in  \mathcal{K}_{t',n}\cup\mathcal{K}_{t',n}\cdot\mathcal{K}_{t',n}}|p(f)-p_n(f)|&\le  \frac{\delta^2}{2\max\{w_{\max}+w_{\min}/2,2/w_{\min}\}},
\label{eq:assume-2}
\end{align}
where $\delta=\frac{w_{\min}}{4w_{\max}+3w_{\min}}, t'=t/18$.
There exists a constant $C$ only depends on $\M$ and kernel function $R$, such that
  \begin{eqnarray}
    \|(T_{t,n}-T_t)T_{t,n}\|_{C^1}\le \frac{Ch_0}{t^{3k/4+3/2}},\quad \|(T_{t,n}-T_t)f\|_{C^1}\le \frac{Ch(f)}{t^{3k/4+3/2}}.\nonumber
  \end{eqnarray}
where 
\begin{eqnarray}
\label{eq:h}
h_0&=&
\sup_{g\in \mathcal{R}_t\cdot \mathcal{K}_{t,n}\cup  \mathcal{R}_t }|p_n(g)-p(g)|+ t\sup_{g\in \mathcal{D}_t\cup
\overline{\mathcal{K}}_{t,n}\cdot \overline{\mathcal{R}}_t\cup
\mathcal{K}_{t,n}\cdot \overline{\mathcal{R}}_t\cup \mathcal{K}_{t,n}\cdot \mathcal{D}_t}|p_n(g)-p(g)|
\\
&&+t^2\sup_{g\in  \mathcal{K}_{t,n}\cdot \overline{\mathcal{D}}_t}|p_n(g)-p(g)|+t^3\sup_{g\in  \overline{\mathcal{K}}_{t,n}\cdot \overline{\mathcal{D}}_t}|p_n(g)-p(g)|,\nonumber
\\
\label{eq:hf}
h(f)&=&
\sup_{g\in \mathcal{R}_t\cdot \mathcal{K}_{t,n}\cup  \mathcal{R}_t }|p_n(g)-p(g)|+ t\sup_{g\in \mathcal{D}_t\cup
f\cdot \overline{\mathcal{R}}_t\cup
\mathcal{K}_{t,n}\cdot \overline{\mathcal{R}}_t\cup \mathcal{K}_{t,n}\cdot \mathcal{D}_t}|p_n(g)-p(g)|
\\
&&+t^2\sup_{g\in  \mathcal{K}_{t,n}\cdot \overline{\mathcal{D}}_t}|p_n(g)-p(g)|+t^3\sup_{g\in  \overline{\mathcal{K}}_{t,n}\cdot \overline{\mathcal{D}}_t}|p_n(g)-p(g)|,\nonumber
\end{eqnarray}  
\end{theorem}
The proof of this theorem will be deferred to Section \ref{sec:converge-c1}.

%%% Local Variables:
%%% mode: latex
%%% TeX-master: "paper_consistent"
%%% End:

%\input{intermediate_robin}

\subsection{Entropy bound}

In this subsection, we will verify the assumption \eqref{eq:assume-1}, \eqref{eq:assume-2} in Theorem \ref{thm:converge_c1} and estimate $h_0$ and $h(f)$ 
defined in \eqref{eq:h} and \eqref{eq:hf} to get the convergence rate. The method we use is to estimate the covering number of function classes defined in 
previous subsection. First we introduce the definition of covering number.

Let $(Y, d)$ be a metric space and set $F\subset Y$ . For every $\e>0$, denote by $N (\e, F, d)$ the
minimal number of open balls (with respect to the metric $d$) needed to cover $F$. That is,
the minimal cardinality of the set $\{y_1 , \cdots, y_m\}\subset Y$ with the property that every $f \in  F$
has is some $y_i$ such that $d(f, y_i ) < \e$. The set $\{y_1 , \cdots, y_m \}$ is called an $\e$-cover of $F$ . The
logarithm of the covering numbers is called the entropy of the set.
For every sample $\{x_1 , \cdots, x_n \}$
let $\mu_n$ be the empirical measure supported on that sample. For $1 \le p <\infty$ and a
function $f$ , put $\|f\|_{L_p (\mu_n )} =\left(\frac{1}{n}\sum_{i=1}^n  |f (x_i )|^p\right)^{1/p}$
and set $\|f\|_\infty = \max_{1\le i\le n} |f (x_i )|$. Let
$N(\e , F, L_p (\mu_n )$ be the covering numbers of $F$ at scale $\e$ with respect to the $L_p (\mu_n )$ norm.

We will use following theorem which is well known in empirical process theory.
\begin{theorem}(Theorem 2.3 in \cite{entropy})
\label{thm:entropy}
Let $F$ be a class of functions from $\M$ to $[-1,1]$ and set $\mu$ to be a probability measure on $\M$. 
Let $(\bx_i)_{i=1}^\infty$ be independent 
random variables distributed according to $\mu$. For every $\e>0$ and any $n\ge 8/\e^2$,
  \begin{eqnarray}
    \mathbb{P}\left(\sup_{f\in F}|\frac{1}{n}\sum_{i=1}^nf(\bx_i)-\int_\M f(\bx)\mu(\bx)\mathd \bx|>\e\right)
\le 8\mathbb{E}_\mu[ N(\e/8,F,L_1(\mu_n))]\exp(-n\e^2/128)
  \end{eqnarray}
\end{theorem}
Notice that 
  \begin{align*}
    L_1(\mu_n)\le L_\infty(\mu_n)\le L_\infty
  \end{align*}
where $\|f\|_{L_\infty}=\max_{\bx\in\M}|f(\bx)|$. Then we get following corollary.
\begin{corollary}
\label{cor:entropy-0}
Let $F$ be a class of functions from $\M$ to $[-1,1]$ and set $\mu$ to be a probability measure on $\M$. 
Let $(\bx_i)_{i=1}^\infty$ be independent 
random variables distributed according to $\mu$. For every $\e>0$ and any $n\ge 8/\e^2$,
  \begin{eqnarray}
    \mathbb{P}\left(\sup_{f\in F}|\frac{1}{n}\sum_{i=1}^nf(\bx_i)-\int_\M f(\bx)\mu(\bx)\mathd \bx|>\e\right)
\le 8 N(\e/8,F,L_\infty)\exp(-n\e^2/128)
  \end{eqnarray}
where $N(\e , F, L_\infty)$ be the covering numbers of $F$ at scale $\e$ with respect to the $L_\infty$ norm
\end{corollary}

% \begin{align}
%   \label{eq:note-mc}
%   p(f)=\int_
% \end{align}
\begin{corollary}
\label{cor:entropy}
Let $F$ be a class of functions from $\M$ to $[-1,1]$. Let $(\bx_i)_{i=1}^\infty$ be independent 
random variables distributed according to $p$, where $p$ is the probability distribution in Assumption \ref{assumptions}. Then
with probability at least $1-\delta$,
  \begin{eqnarray}
    \sup_{f\in F}|p(f)-p_n(f)|\le  \sqrt{\frac{128}{n}\left(\ln N(\sqrt{\frac{2}{n}},F,L_\infty)
+\ln \frac{8}{\delta}\right)},\nonumber
  \end{eqnarray}
where 
\begin{align}
  \label{eq:note-mc}
  p(f)=\int_\M f(\bx)p(\bx)\mathd \bx,\quad p_n(f)=\frac{1}{n}\sum_{i=1}^nf(\bx_i).
\end{align}
\end{corollary}
\begin{proof}
  Using Corollary \ref{cor:entropy-0}, with probability at least $1-\delta$, 
\begin{eqnarray}
    \sup_{f\in F}|p(f)-p_n(f)|\le  \e_\delta,\nonumber
  \end{eqnarray}
where $\e_\delta$ is determined by
\begin{eqnarray}
  \e_\delta=\sqrt{\frac{128}{n}\left(\ln N(\e_\delta/8,F,L_\infty)+\ln \frac{8}{\delta}\right)}.\nonumber
\end{eqnarray}
Obviously, 
\begin{align*}
  \e_\delta\ge \sqrt{\frac{128}{n}}=8\sqrt{\frac{2}{n}}
\end{align*}
which gives that
\begin{align*}
  N(\e_\delta/8,F,L_\infty)\le N(\sqrt{\frac{2}{n}},F,L_\infty)
\end{align*}
Then, we have
\begin{eqnarray}
  \e_\delta \le \sqrt{\frac{128}{n}\left(\ln N(\sqrt{\frac{2}{n}},F,L_\infty)
+\ln \frac{8}{\delta}\right)}\nonumber
\end{eqnarray}
which proves the corollary.
\end{proof}
Above corollary provides a tool to estimate the integral error on random samples. The key point is to 
obtain the estimates of the covering number.% $N(\sqrt{\frac{2}{n}},F,L_\infty)$ with different function set $F$. 

Let us start from the function class $\mathcal{R}_t$. The functions in $\mathcal{R}_t$ are bounded uniformly, and the bound only depends on the kernel function $R$.
To apply above corollary, we need to normalize $\mathcal{R}_t$ to make it lie in $[-1,1]$. Here we also use $\mathcal{R}_t$ to denote the normalized function class 
and absorb the bound of $\mathcal{R}_t$ into the generic constant $C$. We do same normalize procedure for all function classes defined in Section \ref{sec:intermediate}.

Since the kernel $R\in C^2(\M)$ and $\M\in C^\infty$, we have
for any $\bx,\by\in \M$
\begin{eqnarray}
  |R\left(\frac{\|\bx-\by\|^2}{4t}\right)-R\left(\frac{\|\bz-\by\|^2}{4t}\right)|\le \frac{C}{\sqrt{t}}\|\bx-\bz\|.\nonumber
\end{eqnarray}
This gives an easy bound of $N(\e, \mathcal{R}_t,L_\infty)$,
\begin{eqnarray}
\label{eq:bound-R}
  N(\e, \mathcal{R}_t,L_\infty)\le \left(\frac{C}{\e\sqrt{t}}\right)^k
\end{eqnarray}
% Similarly, 
% \begin{eqnarray}
% \label{eq:bound-D}
%   N(\e, \mathcal{D}_t,L_\infty)\le \left(\frac{C}{\e\sqrt{t}}\right)^k
% \end{eqnarray}
Using Corollary \ref{cor:entropy}, with probability at least $1-1/(2n)$,
\begin{align}
\label{eq:bound-w-prob}
  \sup_{f\in \mathcal{R}_t\cup\mathcal{R}_{t'}\cup\mathcal{R}_{8t}}|p(f)-p_n(f)|\le \frac{C}{\sqrt{n}}\left(\ln n-\ln t +1\right)^{1/2}
\end{align}
%Here $C$ is a constant depends on $\M$ and the kernel function $R$.%, $k$ is the dimension of $\M$.
Then, we have
\begin{corollary}
\label{cor:prob_bound_wt}
  With probability at least $1-1/(2n)$, 
  \begin{align*}
\sup_{f\in \mathcal{R}_t\cup\mathcal{R}_{t'}\cup\mathcal{R}_{8t}}|p(f)-p_n(f)|\le \frac{w_{\min}}{2}
  \end{align*}
as long as $n$ is large enough such that the right hand side of \eqref{eq:bound-w-prob} is less than $w_{\min}/2$.%, where $r_0>0$ is the constant 
%in Lemma \ref{lem:bound_w_t_h}.  
\end{corollary}
To get the covering number $ N(\e,\mathcal{K}_{t,n},L_\infty)$, we need the assumption that $\sup_{f\in \mathcal{R}_t}|p(f)-p_n(f)|\le \frac{w_{\min}}{2}$.
\begin{align*}
  \left|\frac{1}{w_{t,n}(\by)}\left[R\left(\frac{\|\bx-\by\|^2}{4t}\right)-R\left(\frac{\|\bz-\by\|^2}{4t}\right)\right]\right|\le
&\frac{2}{w_{\min}}|R\left(\frac{\|\bx-\by\|^2}{4t}\right)-R\left(\frac{\|\bz-\by\|^2}{4t}\right)|
\le \frac{C}{\sqrt{t}}|\bx-\by|
\end{align*}
The first inequality comes from the fact that $\min_{\mathbf{z}\in \M} w_{t,n}(\mathbf{z})\ge w_{\min}/2$ which is guaranteed by the assumption that 
$\sup_{f\in \mathcal{R}_t}|p(f)-p_n(f)|\le \frac{w_{\min}}{2}$.
Then we have
\begin{eqnarray}
  N(\e,\mathcal{K}_{t,n},L_\infty)\le \left(\frac{C}{\e\sqrt{t}}\right)^{k}.
\end{eqnarray}
%if $\sup_{f\in \mathcal{R}_t}|p(f)-p_n(f)|\le \frac{w_{\min}}{2}$. 
Similarly, we can get %the estimate of $N(\e,\mathcal{K}_{t,n}\cdot\mathcal{K}_{t,n},L_\infty)$
\begin{eqnarray}
  N(\e,\mathcal{K}_{t,n}\cdot\mathcal{K}_{t,n},L_\infty)\le \left(\frac{C}{\e\sqrt{t}}\right)^{2k}
\end{eqnarray}
%under the assumption that $\sup_{f\in \mathcal{R}_t}|p(f)-p_n(f)|\le \frac{w_{\min}}{2}$.

Using Corollary \ref{cor:entropy}, if $\sup_{f\in \mathcal{R}_t}|p(f)-p_n(f)|\le \frac{w_{\min}}{2}$, then
\begin{align}
%\label{eq:bound-w-prob}
  \sup_{f\in \mathcal{K}_{t,n}\cup \mathcal{K}_{t,n}\cdot \mathcal{K}_{t,n}}|p(f)-p_n(f)|\le C\sqrt{ \frac{k}{n}}\left(\ln n-\ln t +1\right)^{1/2}
\end{align}
with probability at least $1-1/(2n)$.
%Here $C$ is a constant depends on $\M$ and the kernel function $R$, $k$ is the dimension of $\M$.
From Corollary \ref{cor:prob_bound_wt}, we know that the assumption $\sup_{f\in \mathcal{R}_t}|p(f)-p_n(f)|\le \frac{w_{\min}}{2}$ holds with probability 
at least $1-1/(2n)$. By integrating these results together, we obtain
\begin{corollary}
  With probability at least $1-1/n$, 
  \begin{align*}
\sup_{f\in \mathcal{K}_{t,n}\cup \mathcal{K}_{t,n}\cdot \mathcal{K}_{t,n}}|p(f)-p_n(f)|\le \frac{\delta^2}{2\max\{w_{\max}+w_{\min}/2,2/w_{\min}\}}
  \end{align*}
as long as $n$ is large enough. Here $\delta=\frac{w_{\min}}{4w_{\max}+3w_{\min}}$.
\end{corollary}
Using similar techniques, we can get the estimate of $h_0$ and $h(f)$ in \eqref{eq:h} and \eqref{eq:hf}. Together with Theorem \ref{thm:converge_h1}, we get
\begin{theorem}
\label{thm:estimate-h}
Let $\phi$ be an eigenfunction of $T$.
  With probability at least $1-1/n$, 
  \begin{align*}
\|(T_t-T_{t,n})T_{t,n}\|_{C^1}\le& \frac{C}{t^{3k/4+3/2}\sqrt{n}}\left(\ln n-\ln t +1\right)^{1/2},\\
\|(T_t-T_{t,n})\phi\|_{C^1}\le& \frac{C_\phi}{t^{3k/4+3/2}\sqrt{n}}\left(\ln n-\ln t +1\right)^{1/2}
  \end{align*}
as long as $n$ is large enough. Here 
$C_\phi$ is a constant depends on $\M$, kernel function $R$, distribution $p$ and eigenfunction $\phi$.
\end{theorem}

% \begin{theorem}
%   With probability at least $1-\delta$,
%   \begin{eqnarray}
%     \sup_{f\in F}|p(f)-p_n(f)|\le \frac{1}{\sqrt{n}}\left(d\ln n+d(2k+1)\ln \frac{1}{t}+\ln\frac{8}{\delta}+C\right)^{1/2}
%   \end{eqnarray}
% \end{theorem}

% \subsection{Proof of Theorem \ref{thm:regularity_boundary}}
% \input{stability.tex}

% \subsection{Proof of Theorem \ref{thm:integral_error}}
% \label{sec:integral_error}
% \input{integral_eqn}

\section{Proof of Theorem \ref{thm:converge_c1}}
\label{sec:converge-c1}

To prove Theorem \ref{thm:converge_c1}, we need following two theorems.
\begin{theorem}
Under the assumption in Assumption \ref{assumptions} and assume \eqref{eq:assume-1}, \eqref{eq:assume-2} hold.
% \begin{align}
% \label{eq:assume-1}
%   \sup_{f\in \mathcal{R}_{t'}\cup\mathcal{R}_{t}}|p(f)-p_n(f)|&\le w_{\min}/2,\\
% \sup_{f\in  \mathcal{K}_{t',n}\cup\mathcal{K}_{t',n}\cdot\mathcal{K}_{t',n}}|p(f)-p_n(f)|&\le  \frac{\delta^2}{2\max\{w_{\max}+w_{\min}/2,2/w_{\min}\}},
% \label{eq:assume-2}
% \end{align}
% where $\delta=\frac{w_{\min}}{4w_{\max}+3w_{\min}}$.
There exist constants $C>0$ only depends on $\M$ and kernel function $R$, so that for
any ${\bf u} = (u_1, \cdots, u_n)^t \in \R^d$ with $\sum_{i=1}^n u_i = 0$,
\begin{equation}
\frac{1}{n^2t}\sum_{i,j=1}^nR_t(\bfp_i,\bfp_j)(u_i-u_j)^2 \geq \frac{C}{n}\sum_{i=1}^nu_i^2.
\end{equation}
% where $\left<{\bf u}, {\bf v}\right> = \sum_{i=1}^n u_iv_i$ for any ${\bf u} = (u_1, \cdots, u_n)^t,
% {\bf v} =(v_1, \cdots, v_n)^t \in \R^d$.
\label{thm:elliptic_dis}
\end{theorem}
The proof of this theorem can be found in Appendix.
% Let $\delta=\frac{w_{\min}}{4w_{\max}+3w_{\min}}$. If $\frac{1}{n}\sum_{i=1}^nu^2(\bfp_i)\ge \frac{\delta^2}{n} \sum_{i=1}^nu_i^2$, and 
% \begin{align*}
%   \sup_{f\in \mathcal{K}_{t',n}\cdot\mathcal{K}_{t',n}\cup \mathcal{K}_{t',n}}|p(f)-p_n(f)|\le \frac{\delta^2}{2\max\{w_{\max}+w_{\min}/2,2/w_{\min}\}}
% \end{align*}

\begin{theorem}
Suppose $\mathbf{u}=(u_1, \cdots, u_n)^t$ with $\sum_i u_i = 0$ solves the problem~\eqref{eqn:dis}
and  $f\in C(\mathcal{M})$.
%$\mathbf{u}=(u_i)$ solves the equation
%\begin{eqnarray}
  %-\frac{1}{t}\sum_{j}R_t(\bfp_i,\bfp_j)(u_i-u_j)V_j+
%\sum_{j}R_t(\bfp_i,\bfp_j)g(\bfp_j)A_j=\sum_{j}R_t(\bfp_i,\bfp_j)f(\bfp_j)V_j
%\end{eqnarray}
Then there exists a constant $C>0$ only depends on $\M$ and kernel function $R$, such that
\begin{eqnarray}
  \left(\frac{1}{n}\sum_{i=1}^nu_i^2\right)^{1/2}\le C\left(\frac{1}{n}\sum_{i=1}^nf(\bfp_i)^2\right)^{1/2}\le C\|f\|_\infty,\nonumber
\end{eqnarray}
as long as \eqref{eq:assume-1}, \eqref{eq:assume-2} are satisfied.
\label{thm:bound_solution_bfu}
\end{theorem}
\begin{proof} 
Since $(u_1, \cdots, u_n)$ satisfies that
  \begin{align*}
    \frac{1}{nt}\sum_{j=1}^nR_t(\bx_i,\bx_j)(u_i-u_j)=\frac{1}{n}\sum_{j=1}^n\bar{R}_t(\bx_i,\bx_j)f(\bx_j)
  \end{align*}
using Theorem \ref{thm:elliptic_dis}, we have
  \begin{align*}
     \frac{C}{n}\sum_{i=1}^nu_i^2\le &\frac{1}{n^2t}\sum_{i,j=1}^nR_t(\bfp_i,\bfp_j)(u_i-u_j)^2=\frac{2}{n^2t}\sum_{i,j=1}^nR_t(\bfp_i,\bfp_j)(u_i-u_j)u_i\\ 
=& \frac{2}{n^2}\sum_{i,j=1}^nR_t(\bx_i,\bx_j)f(\bx_j)u_i\\
\le&\left(\frac{1}{n^2}\sum_{i,j=1}^nR_t(\bx_i,\bx_j)f^2(\bx_j)\right)^{1/2}\left(\frac{1}{n^2}\sum_{i,j=1}^nR_t(\bx_i,\bx_j)u_i^2\right)^{1/2}\\
\le& C\left(\frac{1}{n}\sum_{j=1}^nf^2(\bx_j)\right)^{1/2}\left(\frac{1}{n}\sum_{i=1}^nu_i^2\right)^{1/2}\\
\le&C\|f\|_\infty\left(\frac{1}{n}\sum_{i=1}^nu_i^2\right)^{1/2}
  \end{align*}
\end{proof}

\begin{theorem}(\cite{SS-neumann, SS-iso})
Under the assumptions in Assumption \ref{assumptions},
assume $u(\bx)$ solves the following equation
\begin{eqnarray}
  -L_t u = r,
\end{eqnarray}
where
\begin{align}
  \label{eq:Lt}
  L_t u = \frac{C_t}{t}\int_\M R\left(\frac{|\bx-\by|^2}{4t}\right)(u(\bx)-u(\by))p(\by)\mathd\by.
\end{align}
Then, there exist constants $C>0, T_0>0$ independent on $t$, such that
\begin{eqnarray}
  \|u\|_{L^2(\M)}\le C\|r\|_{L^2(\M)}.
\end{eqnarray}
as long as $t\le T_0$.
\label{thm:stable-l2}
\end{theorem}
The proof of above theorem can be found in \cite{SS-neumann}.
\begin{theorem}
Under the assumptions in Assumption \ref{assumptions}.
Let
$f\in C(\M)$
in both problems, then there exists constants
$C>0$, so that
\begin{eqnarray}
\|(T_{t,n}-T_t)T_{t,n} f\|_{L^2(\M)} &\leq& \frac{C}{t^{k/2+1}}\|f\|_{\infty}\left(
\sup_{g\in \mathcal{R}_t \cup \mathcal{R}_t\cdot \mathcal{K}_{t,n}}|p_n(g)-p(g)|+ t\sup_{g\in \overline{\mathcal{K}}_{t,n}\cdot \overline{\mathcal{R}}_t\cup
{\mathcal{K}}_{t,n}\cdot \overline{\mathcal{R}}_t}|p_n(g)-p(g)|\right)\nonumber\\
\|(T_{t,n}-T_t)f\|_{L^2(\M)} &\leq& \frac{C}{t^{k/2+1}}\|f\|_{\infty}\left(
\sup_{g\in   \mathcal{R}_t \cup\mathcal{R}_t\cdot \mathcal{K}_{t,n}}|p_n(g)-p(g)|+
t\sup_{g\in {\mathcal{K}}_{t,n}\cdot \overline{\mathcal{R}}_t\cup f\cdot \overline{\mathcal{R}}_t}|p_n(g)-p(g)|\right),\nonumber
\end{eqnarray}
as long as $t$ small enough and \eqref{eq:assume-1}, \eqref{eq:assume-2} are satisfied. 
\label{thm:dis_error}
\end{theorem}

\begin{proof} {\it of Theorem \ref{thm:dis_error}} \\
First, denote
\begin{eqnarray}
u_{t,n}(\bx)=T_{t,n}f=\frac{1}{n\,w_{t,n}(\bx)}\left(\sum_{j=1}^nR_t(\bx,\bfp_j)u_j
-t\sum_{j=1}^n\bar{R}_t(\bx,\bfp_j)f_j\right)
\end{eqnarray}
where $\mathbf{u}=(u_1, \cdots, u_n)^t$ with $\sum_{i=1}^n u_i = 0$ solves the problem~\eqref{eqn:dis},
$f_j = f(\bfp_j)$ and $w_{t,n}(\bx)=\frac{1}{n}\sum_{j=1}^nR_t(\bx,\bfp_j)$.
And denote
\begin{eqnarray}
v_{t,n}(\bx)=T_{t,n}u_{t,n}=\frac{1}{n\,w_{t,n}(\bx)}\left(\sum_{j=1}^nR_t(\bx,\bfp_j)v_j
-t\sum_{j=1}^n\bar{R}_t(\bx,\bfp_j)u_j\right)
\end{eqnarray}
where $\mathbf{v}=(v_1, \cdots, v_n)^t$ with $\sum_{i=1}^n v_i = 0$ solves
\begin{eqnarray}
\label{eqn:dis-v}
 -\frac{1}{nt}\sum_{j=1}^nR_t(\bx_i,\bx_j)(v_i-v_j)=\frac{1}{n}\sum_{j=1}^n\bar{R}_t(\bx_i,\bx_j)u_j.
\end{eqnarray}
It follows from Theorem \ref{thm:elliptic_dis} that there exists a constant $C>0$ independent on $t$ and $n$ such that
\begin{eqnarray}
  \label{eq:bound-uv}
  \left(\frac{1}{n}\sum_{i=1}^nu_i^2\right)^{1/2}\le C\|f\|_{\infty},\quad \left(\frac{1}{n}\sum_{i=1}^nv_i^2\right)^{1/2}\le C\left(\frac{1}{n}\sum_{i=1}^nu_i^2\right)^{1/2}\le
C\|f\|_{\infty}
\end{eqnarray}
The idea to prove the theorem is using Theorem \ref{thm:stable-l2}. Then we need to estimate $\|L_t(T_{t,n}-T_t)T_{t,n}f\|_2$ and $\|L_t(T_{t,n}-T_t)f\|_2$ for any $f\in C(\M)$.

For any $f\in C(\M)$,
\begin{eqnarray}
  L_t(T_{t,n}-T_t)T_{t,n}f&=&\left(L_tT_{t,n}T_{t,n}f-L_{t,n}T_{t,n}T_{t,n}f\right)+\left(L_{t,n}T_{t,n}T_{t,n}f-L_tT_{t}T_{t,n}f\right)
\nonumber\\
&=& (L_tv_{t,n}-L_{t,n}v_{t,n})+\left(L_{t,n}T_{t,n}u_{t,n}-L_tT_{t}u_{t,n}\right).
\label{eq:error-0}
\end{eqnarray}

Next, we estimate two terms of right hand side of \eqref{eq:error-0} separately.
For convenience, we split $v_{t,n}=a_{t,n}+b_{t,n}$ and
\begin{eqnarray}
  a_{t,n}(\bx)&=&\frac{1}{n\,w_{t,n}(\bx)}\sum_{j=1}^nR_t(\bx,\bfp_j)v_j,\\
b_{t,n}(\bx)&=&-\frac{t}{n\,w_{t,n}(\bx)}\sum_{j=1}^n\bar{R}_t(\bx,\bfp_j)u_j.
\end{eqnarray}
% Then
% \begin{eqnarray}
%   L_t(v_{t,n}) - L_{t,n}(v_{t,n})=\left(L_t - L_{t,n}\right)(a_{t,n})+\left(L_t - L_{t,n}\right)(b_{t,n})
% \end{eqnarray}
For $\|L_tb_{t,n} - L_{t,n}b_{t,n}\|_2$, we have
\begin{eqnarray}
&&\left|\left(L_tb_{t,n} - L_{t,n}b_{t,n}\right)(\bx)\right|\nonumber \\
&=& \frac{1}{t} \left|\int_{\mathcal{M}}R_t(\bx,\by)(b_{t,n}(\bx)-b_{t,n}(\by))  p(\by)\mathd \by-\frac{1}{n}
\sum_{j=1}^nR_t(\bx,\bfp_j)
(b_{t,n}(\bx)-b_{t,n}(\bfp_j))\right|\nonumber\\
&\le &\frac{1}{t} \left|b_{t,n}(\bx)\right|\left|\int_{\mathcal{M}}R_t(\bx,\by) p(\by)\mathd \by-\frac{1}{n}
\sum_{j=1}^nR_t(\bx,\bfp_j)\right|\nonumber\\
&&+ \frac{1}{t}\left|\int_{\mathcal{M}}R_t(\bx,\by)b_{t,n}(\by)  p(\by)\mathd \by-\frac{1}{n}
\sum_{j=1}^nR_t(\bx,\bfp_j)b_{t,n}(\bfp_j)\right|
\label{eq:error-b-0}
\end{eqnarray}
The first term of \eqref{eq:error-b-0} can be bounded as following,
\begin{eqnarray}
&&  \left\|b_{t,n}(\bx)\left(\int_{\mathcal{M}}R_t(\bx,\by) p(\by)\mathd \by-\frac{1}{n}\sum_{j=1}^nR_t(\bx,\bfp_j)
\right)\right\|
_{L^2}
\le C_t \|b_{t,n}\|_{L^2} \sup_{g\in \mathcal{R}_t}|p_n(g)-p(g)|%\nonumber\\
%&\le & Ct\|f\|_{\infty} \sup_{g\in \mathcal{R}_t}|p_n(g)-p(g)|
\label{eq:error-b-1}
\end{eqnarray}
and
\begin{eqnarray}
  \|b_{t,n}\|_{L^2}^2&=&\frac{t^2}{n^2}\int_\M \left(\frac{1}{w_{t,n}(\bx)}\sum_{j=1}^n\bar{R}_t(\bx,\bfp_j)u_j\right)^2p(\bx)
\mathd\bx\nonumber\\
&\le & \frac{Ct^2}{n}\int_\M \left(\frac{1}{n}\sum_{j=1}^n\bar{R}_t(\bx,\bfp_j)\right)\left(
\sum_{j=1}^n\bar{R}_t(\bx,\bfp_j)u_j^2\right)
p(\bx)\mathd\bx\nonumber\\
&\le&\frac{Ct^2}{n}
\sum_{j=1}^n\left(u_j^2\int_\M \bar{R}_t(\bx,\bfp_j)p(\bx)\mathd\bx\right)\nonumber\\
&\le& \frac{Ct^2}{n}
\sum_{j=1}^nu_j^2\le Ct^2\|f\|_{\infty},
\end{eqnarray}
where last inequality comes from \eqref{eq:bound-uv}.

%$|b_{t,n}(\bx)|\le Ct\|f\|_{\infty}$ with probability $P(\sup_{g\in \mathcal{R}_t\cup\overline{\mathcal{R}}_t}|p_n(g)-p(g)|\le C_0)$
For the second term of \eqref{eq:error-b-0},
\begin{eqnarray}
 && \left|\int_{\mathcal{M}}R_t(\bx,\by)b_{t,n}(\by)  p(\by)\mathd \by-\frac{1}{n}\sum_{j=1}^nR_t(\bx,\bfp_j)b_{t,n}(\bfp_j)\right|
\nonumber\\
&= &\frac{t}{n} \left|\int_{\mathcal{M}}\frac{R_t(\bx,\by)}{w_{t,n}(\by)}\left(\sum_{\bfp_k\in P}\bar{R}_t(\by,\bfp_k)u_k\right)
  p(\by)
\mathd \by-\frac{1}{n}\sum_{j=1}^n\frac{R_t(\bx,\bfp_j)}{w_{t,n}(\bfp_j)}\sum_{\bfp_k\in P}\bar{R}_t(\bfp_j,\bfp_k)
u_k\right|
\nonumber\\
&\le &\frac{t}{n} \sum_{k=1}^n|u_k|\left|\int_{\mathcal{M}}\frac{R_t(\bx,\by)}{w_{t,n}(\by)}\bar{R}_t(\by,\bfp_k)  p(\by)
\mathd \by-\frac{1}{n}\sum_{j=1}^n\frac{R_t(\bx,\bfp_j)}{w_{t,n}(\bfp_j)}\bar{R}_t(\bfp_j,\bfp_k)\right|
% \nonumber\\
% &\le & t\left(\frac{1}{n}\sum_{j=1}^nu_j^2\right)^{1/2} \sup_{g\in \mathcal{K}_{t,n}\cdot \bar{\mathcal{R}}_t}|p_n(g)-p(g)|\nonumber\\
% &\le & Ct\|f\|_{\infty} \sup_{g\in \mathcal{K}_{t,n}\cdot \bar{\mathcal{R}}_t}|p_n(g)-p(g)|
\label{eq:error-b-1-sum}
\end{eqnarray}
Let
\begin{eqnarray}
A &=&  C_t\int_{\mathcal{M}}\frac{1}{w_{t,n}(\by)}R\left(\frac{|\bx-\by|^2}{4t}\right)
\bar{R}\left(\frac{|\bfp_i-\by|^2}{4t}\right) p(\by)\mathd \by\nonumber\\
 &-&\frac{C_t}{n}\sum_{j=1}^n \frac{1}{w_{t,n}(\bfp_j)}R\left(\frac{|\bx-\bfp_j|^2}{4t}\right)\bar{R}\left(\frac{|\bfp_i-\bfp_j|^2}{4t}\right).
\end{eqnarray}
We have
\begin{eqnarray}
  |A|<C_t\sup_{g\in \mathcal{K}_{t,n}\cdot\overline{\mathcal{R}}_t}|p_n(g)-p(g)|
\end{eqnarray}
for some constant $C$ independent of $t$. In addition, notice that
only when $|\bx-\bfp_i|^2\leq 16t $ is $A\neq 0$, which implies
\begin{eqnarray}
|A| \leq \frac{1}{\delta_0}|A|R\left(\frac{|\bx-\bfp_i|^2}{32t}\right).
\end{eqnarray}
Using these properties of $A$, we obtain
\begin{eqnarray}
 && \left|\int_{\mathcal{M}}R_t(\bx,\by)b_{t,n}(\by)  p(\by)\mathd \by-\frac{1}{n}\sum_{j=1}^nR_t(\bx,\bfp_j)b_{t,n}(\bfp_j)\right|
\nonumber\\
&\le&\frac{Ct}{n}|A|_\infty \sum_{k=1}^n|u_k|R\left(\frac{|\bx-\bfp_k|^2}{32t}\right)
\nonumber\\
&\le & \frac{Ct}{n}\sum_{k=1}^nC_t|u_k|R\left(\frac{|\bx-\bfp_k|^2}{32t}\right)
C_t\sup_{g\in \mathcal{K}_{t,n}\cdot \overline{\mathcal{R}}_t}|p_n(g)-p(g)|
\end{eqnarray}
It follows that
\begin{eqnarray}
 && \left\|\int_{\mathcal{M}}R_t(\bx,\by)b_{t,n}(\by)  p(\by)\mathd \by-\frac{1}{n}\sum_{j=1}^nR_t(\bx,\bfp_j)b_{t,n}(\bfp_j)\right\|_2
\nonumber\\
&\le & Ct\left(\int_\M\left(\frac{1}{n}\sum_{k=1}^nC_t|u_k|R\left(\frac{|\bx-\bfp_k|^2}{32t}\right)\right)^2p(\bx)\mathd\bx\right) ^{1/2}
C_t\sup_{g\in \mathcal{K}_{t,n}\cdot \overline{\mathcal{R}}_t}|p_n(g)-p(g)|\nonumber\\
&\le & Ct \left(\frac{1}{n}\sum_{k=1}^nu_k^2\right)^{1/2}C_t\sup_{g\in \mathcal{K}_{t,n}\cdot \overline{\mathcal{R}}_t}|p_n(g)-p(g)|\nonumber\\
&\le & Ct\|f\|_{\infty} C_t\sup_{g\in \mathcal{K}_{t,n}\cdot \overline{\mathcal{R}}_t}|p_n(g)-p(g)|
\label{eq:error-b-2}
\end{eqnarray}
To get the second inequality, we use the condtion that $C_t\sup_{g\in \mathcal{R}_{8t}}\le w_{\min}/2$.

Now we have complete upper bound of $\|L_tb_{t,n} - L_{t,n}b_{t,n}\|_{L_2}$ using \eqref{eq:error-b-0}, \eqref{eq:error-b-1} and \eqref{eq:error-b-2} 
and $C_t=\frac{1}{(4\pi t)^{k/2}}$,
\begin{eqnarray}
&&\|L_tb_{t,n} - L_{t,n}b_{t,n}\|_{L^2(\M)}
 \le
\frac{C}{t^{k/2}}\|f\|_{\infty}\left(\sup_{g\in \mathcal{R}_t\cup\mathcal{K}_{t,n}\cdot\overline{\mathcal{R}}_t}
|p_n(g)-p(g)|\right).
\label{eq:error-b}
\end{eqnarray}
Mimicing the derivation of \eqref{eq:error-b}, we have
\begin{eqnarray}
&&\|L_ta_{t,n} - L_{t,n}a_{t,n}\|_{L^2(\M)} % \nonumber \\
% &=&\left(\int_M \left|\left(L_t(a_{t,n}) - L_{t,n}(a_{t,n})\right)(\bx)\right|^2 \mathd\bx\right)^{1/2} \nonumber\\
% %&=&  \left(\int_{\mathcal{M}}\left|\int_{\mathcal{M}}R_t(\bx,\by)(u_{t,n}^1(\bx)-u_{t,n}^1(\by))  p(\by)\mathd \by-\sum_{j}R_t(\bx,\bfp_j)(u_{t,n}^1(\bx)-u_{t,n}^1(\bfp_j))\right|^2\mathd\bx \right)^{1/2}\nonumber\\
% &\le & \frac{1}{t}\left(\int_{\mathcal{M}}\left(a_{t,n}(\bx)\right)^2\left|\int_{\mathcal{M}}R_t(\bx,\by) p(\by)\mathd \by-\sum_{j=1}^nR_t(\bx,\bfp_j)\right|^2\mathd\bx\right)^{1/2} \nonumber\\
% &&+  \frac{1}{t}\left(\int_{\mathcal{M}}\left|\int_{\mathcal{M}}R_t(\bx,\by)a_{t,n}(\by)  p(\by)\mathd \by-\sum_{j=1}^nR_t(\bx,\bfp_j)a_{t,n}(\bfp_j)\right|^2\mathd\bx
%  \right)^{1/2} \nonumber \\
%  &\le &
\le
\frac{C}{t^{k/2+1}}\|f\|_{\infty}\left(\sup_{g\in \mathcal{R}_t\cup\mathcal{K}_{t,n}\cdot\mathcal{R}_t}
|p_n(g)-p(g)|\right)
\end{eqnarray}
And consequently,
%Now assembling the parts together, we have the following upper bound for $\|L_tv_{t,n} - L_{t,n}v_{t,n}\|_{L^2(\M)}$.
\begin{eqnarray}
&&\|L_tv_{t,n} - L_{t,n}v_{t,n}\|_{L^2(\M)}\nonumber \\
&\le& \|L_ta_{t,n} - L_{t,n}a_{t,n}\|_{L^2(\M)} + \|L_tb_{t,n} - L_{t,n}b_{t,n}\|_{L^2(\M)}  \nonumber \\
&\le&\frac{C}{t^{k/2+1}}\|f\|_{\infty}\left(\sup_{g\in \mathcal{R}_t\cup\mathcal{K}_{t,n}\cdot\mathcal{R}_t}
|p_n(g)-p(g)|+t\sup_{g\in  \mathcal{K}_{t,n}\cdot\overline{\mathcal{R}}_t}
|p_n(g)-p(g)|\right).
\label{eq:error-dis-v}
\end{eqnarray}

The second term of \eqref{eq:error-0} can be bounded as following,
\begin{eqnarray}
&&L_t(T_t u_{t,n}) - L_{t, n}(T_{t,n}u_{t, n}) \nonumber \\
%&=&\left(\int_\M \left(\left(L_tu_{t} - L_{t, n}u_{t, n}\right)(\bx)\right)^2     \mathd\bx\right)^{1/2}  \nonumber \\
&\leq&    \int_{\mathcal{M}}\bar{R}_t(\bx,\by)u_{t,n}(\by)p(\by)\mathd\by - \frac{1}{n}
\sum_{j=1}^n\bar{R}_t(\bx,\bfp_j)u_j \nonumber \\
&\le & \frac{1}{n^2}\sum_{j=1}^n \frac{\overline{R}_t(\bx,\bx_j)}{w_{t,n}(\bx_j)}\left(\sum_{k=1}^n
R_t(\bx_j,\bx_k)u_k-t\sum_{k=1}^n \overline{R}_t(\bx_j,\bx_k)f_k\right)\nonumber\\
&&
-\frac{1}{n}\int_\M \frac{\overline{R}_t(\bx,\by)}{w_{t,n}(\by)}\left(\sum_{k=1}^n
R_t(\by,\bx_k)u_k-t\sum_{k=1}^n \overline{R}_t(\by,\bx_k)f_k\right)p(\by)\mathd\by\nonumber\\
&=& \frac{1}{n}\sum_{k=1}^nu_k \left(\frac{1}{n}\sum_{j=1}^n\frac{\overline{R}_t(\bx,\bx_j)}{w_{t,n}(\bx_j)}
R_t(\bx_j,\bx_k)-\int_\M \frac{\overline{R}_t(\bx,\by)}{w_{t,n}(\by)}R_t(\by,\bx_k)p(\by)\mathd\by\right)\nonumber\\
&&-\frac{t}{n}\sum_{k=1}^nf_k \left(\frac{1}{n}\sum_{j=1}^n\frac{\overline{R}_t(\bx,\bx_j)}{w_{t,n}(\bx_j)}
\overline{R}_t(\bx_j,\bx_k)
-\int_\M \frac{\overline{R}_t(\bx,\by)}{w_{t,n}(\by)}\overline{R}_t(\by,\bx_k)p(\by)\mathd\by\right).
%C \sup_{g\in f\cdot \overline{\mathcal{R}}_t}|p_n(g)-p(g)|.
\label{eq:error-dis-u}
\end{eqnarray}
Using the similar derivation from \eqref{eq:error-b-1-sum} to \eqref{eq:error-b}, we get
\begin{eqnarray}
&&  \|L_t(T_t u_{t,n}) - L_{t, n}(T_{t,n}u_{t, n})\|_{L^2}\nonumber\\
&\le &C\left(\frac{1}{n}\sum_{j=1}^nu_j^2\right)^{1/2} C_t\sup_{g\in {\mathcal{K}}_{t,n}\cdot \overline{\mathcal{R}}_t}|p_n(g)-p(g)|
+Ct\|f\|_{\infty}C_t\sup_{g\in \overline{\mathcal{K}}_{t,n}\cdot \overline{\mathcal{R}}_t}|p_n(g)-p(g)|\nonumber\\
&\le& \frac{C}{t^{k/2}}\|f\|_{\infty}\left(\sup_{g\in
{\mathcal{K}}_{t,n}\cdot \overline{\mathcal{R}}_t}|p_n(g)-p(g)|+t\sup_{g\in \overline{\mathcal{K}}_{t,n}\cdot \overline{\mathcal{R}}_t}|p_n(g)-p(g)|\right).
\end{eqnarray}
The complete estimate follows from Equation~\eqref{eq:error-dis-v} and~\eqref{eq:error-dis-u}.
\begin{eqnarray}
\|L_t(T_{t,n}-T_t)T_{t, n}f\|_{L^2(\M)}
&\le&   \frac{C}{t^{k/2+1}}\|f\|_{\infty}\left(
\sup_{g\in \mathcal{R}_t \cup \mathcal{R}_t\cdot \mathcal{K}_{t,n}}|p_n(g)-p(g)|\right.\nonumber\\
&&\left.+ t\sup_{g\in
{\mathcal{K}}_{t,n}\cdot \overline{\mathcal{R}}_t}|p_n(g)-p(g)|+t^2\sup_{g\in \overline{\mathcal{K}}_{t,n}\cdot \overline{\mathcal{R}}_t}|p_n(g)-p(g)|\right).
\label{eq:error-tf}
\end{eqnarray}
Similarly, we can also get
\begin{eqnarray}
\|L_t(T_{t,n}-T_t)f)\|_{L^2(\M)}
&\le& \frac{C}{t^{k/2+1}}\|f\|_{\infty}\left(
\sup_{g\in \mathcal{R}_t \cup \mathcal{R}_t\cdot \mathcal{K}_{t,n}}|p_n(g)-p(g)|\right.\nonumber\\
&&\left.+ t\sup_{g\in
{\mathcal{K}}_{t,n}\cdot \overline{\mathcal{R}}_t}|p_n(g)-p(g)|+t^2\sup_{g\in f\cdot \overline{\mathcal{R}}_t}|p_n(g)-p(g)|\right).
\label{eq:error-f}
\end{eqnarray}
The theorem is proved by using Theorem \ref{thm:stable-l2} and above two estimates \eqref{eq:error-tf}, \eqref{eq:error-f}
% \begin{eqnarray}
% &&\|\nabla\left(L_t(u_{t}) - L_{t,n}(u_{t,n})\right)\|_{L^2(\M)} \nonumber \\
% &\leq&  \left( \int _\M \left( \int_{\mathcal{M}}\nabla_\bx\bar{R}_t(\bx,\by)f(\by) -
% \frac{1}{n}\sum_{j=1}^n\nabla_\bx\bar{R}_t(\bx,\bfp_j)f(\bfp_j)\right)^2 \mathd\bx\right)^{1/2} \nonumber \\
% &\le& \frac{C}{t^{1/2}}\sup_{g\in f\cdot \overline{\mathcal{D}}_t}|p_n(g)-p(g)|
% \label{eqn:du2}
% \end{eqnarray}

% \begin{eqnarray}
% \|\nabla_\bx L_t(u_{t}-u_{t, n})\|_{L^2(\M)}
% \le   \frac{C}{t^{3/2}}\|f\|_{\infty}
% \sup_{g\in \mathcal{D}_t\cdot \mathcal{K}_{t,n}\cup \mathcal{D}_t\cup \mathcal{R}_t }|p_n(g)-p(g)|
% + \frac{C}{t^{1/2}} \sup_{g\in f\cdot \overline{\mathcal{D}}_t}|p_n(g)-p(g)|
% \end{eqnarray}

\end{proof}

 \begin{theorem}
   \label{thm:bound-Ttn}
Under the assumption in Assumption \ref{assumptions} and assume \eqref{eq:assume-1}, \eqref{eq:assume-2} hold.
% \begin{align*}
%   \sup_{f\in \mathcal{R}_{t'}\cup\mathcal{R}_{t}}|p(f)-p_n(f)|&\le w_{\min}/2,\\
% \sup_{f\in  \mathcal{K}_{t',n}\cup\mathcal{K}_{t',n}\cdot\mathcal{K}_{t',n}}|p(f)-p_n(f)|&\le  \frac{\delta^2}{2\max\{w_{\max}+w_{\min}/2,2/w_{\min}\}},
% \end{align*}
% where $\delta=\frac{w_{\min}}{4w_{\max}+3w_{\min}}$.
Then, there exist constants $C>0$ only depends on $\M$ and kernel function $R$, such that for any $f\in C(\M)$,
   \begin{align*}
     \|T_{t,n}f\|_{\infty}\le Ct^{-k/4}\|f\|_\infty, \quad 
\|T_{t,n}f\|_{L^2}\le C\|f\|_\infty.
   \end{align*}
 \end{theorem}
 \begin{proof}
   From the definition of $T_{t,n}$, we have for any $f\in C(\M)$
   \begin{align*}
     T_{t,n}f=\frac{C_t}{n w_{t,n}(\bx)}\sum_{i=1}^n R\left(\frac{|\bx-\bx_i|^2}{4t}\right)u_i+\frac{tC_t }{n w_{t,n}(\bx)}\sum_{i=1}^n R\left(\frac{|\bx-\bx_i|^2}{4t}\right)f(\bx_i)
   \end{align*}
where $(u_1,\cdots,u_n)$ satisfies the equation
\begin{align*}
  \frac{C_t}{n t}\sum_{j=1}^n R\left(\frac{|\bx_i-\bx_j|^2}{4t}\right)(u_i-u_j)=\frac{C_t }{n}\sum_{j=1}^n R\left(\frac{|\bx_i-\bx_j|^2}{4t}\right)f(\bx_j).
\end{align*}
Using Theorem \ref{thm:elliptic_dis}, it is easy to get that
\begin{align*}
  \left(\frac{1}{n}\sum_{i=1}^n u_i^2\right)^{1/2}\le C \|f\|_\infty
\end{align*}
where $C>0$ is a constant only depends on $\M$ and kernel function $R$.

Then
\begin{align*}
  |T_{t,n}f|\le& \left(\frac{C_t}{n w_{t,n}(\bx)}\sum_{i=1}^n R\left(\frac{|\bx-\bx_i|^2}{4t}\right)\right)^{1/2}
\left(\frac{C_t}{n w_{t,n}(\bx)}\sum_{i=1}^n R\left(\frac{|\bx-\bx_i|^2}{4t}\right)u_i^2\right)^{1/2}\\
&+\frac{t C_t }{n w_{t,n}(\bx)}\sum_{i=1}^n R\left(\frac{|\bx-\bx_i|^2}{4t}\right)\|f\|_\infty\\
\le& \left(\frac{C_t}{n w_{t,n}(\bx)}\sum_{i=1}^n R\left(\frac{|\bx-\bx_i|^2}{4t}\right)u_i^2\right)^{1/2}+t\|f\|_\infty\\
\le & \left(\frac{2C_t}{w_{\min}}\right)^{1/2}\left(\frac{1}{n}\sum_{i=1}^n u_i^2\right)^{1/2}+t\|f\|_\infty\le C\|f\|_\infty.
\end{align*}
and 
\begin{align*}
  \|T_{t,n}f\|_{L^2}^2
\le& 2\int_\M \frac{C_t}{n w_{t,n}(\bx)}\sum_{i=1}^n R\left(\frac{|\bx-\bx_i|^2}{4t}\right)u_i^2p(\bx)\mathd \bx+2t^2\|f\|_\infty^2\\
\le & C\left(\frac{1}{n}\sum_{i=1}^n u_i^2+t^2\|f\|_\infty^2\right)\le C\|f\|_\infty^2.
\end{align*}
 \end{proof}
Now, we are ready to prove Theorem \ref{thm:converge_c1}. The main idea is to lift the covergence from $L^2$ to $C^1$ by using the regularity of the kernel function.
 The details are given as following.

\begin{proof} {\it of Theorem \ref{thm:converge_c1}}:

For any $f\in C^1(\M)$, let $u_{t,n}=T_{t,n}f$ and $v_i=T_{t,n}u_{t,n}(\bx_i),\; i=1,\cdots,n$. 
Using the definition of $T_t$ and $T_{t,n}$, $T_tu_{t,n}$ and $T_{t,n}u_{t,n}$ have following representations
  \begin{eqnarray}
    T_tu_{t,n}&=&\frac{1}{w_t(\bx)}\int_{\M}R_t(\bx,\by)T_tu_{t,n}(\by)p(\by)\mathd\by
+\frac{t}{w_t(\bx)}\int_\M\bar{R}(\bx,\by)u_{t,n}(\by)p(\by)\mathd\by,\nonumber\\
T_{t,n}u_{t,n}&=&\frac{1}{n\,w_{t,n}(\bx)}\sum_{i=1}^nR_t(\bx,\bx_i)v_i
+\frac{t}{n\,w_{t,n}(\bx)}\sum_{i=1}^n\bar{R}(\bx,\bx_i)u_i.
\label{eq:Ttn-1}
  \end{eqnarray}
where $u_i=u_{t,n}(\bx_i),\; i=1,\cdots, n$. We know that $(u_1,\cdots,u_n)$ and $(v_1,\cdots,v_n)$ satisfy following equations
respectively
\begin{align*}
&\frac{1}{nt}\sum_{j=1}^nR_t(\bx_i,\bx_j)(u_i-u_j)=\frac{1}{n}\sum_{i=1}^nR_t(\bx_i,\bx_j)f(\bx_j),\\
&  \frac{1}{nt}\sum_{j=1}^nR_t(\bx_i,\bx_j)(v_i-v_j)=\frac{1}{n}\sum_{i=1}^nR_t(\bx_i,\bx_j)u_j.
\end{align*}
Using Theorem \ref{thm:bound_solution_bfu}, we have
\begin{eqnarray}
  \label{eq:bound-uv-2}
  \left(\frac{1}{n}\sum_{i=1}^nu_i^2\right)^{1/2}\le C\|f\|_\infty,\quad \left(\frac{1}{n}\sum_{i=1}^nv_i^2\right)^{1/2}\le C\left(\frac{1}{n}\sum_{i=1}^nu_i^2\right)^{1/2}\le
C\|f\|_\infty
\end{eqnarray}
Denote
  \begin{eqnarray}
    T^1_tu_{t,n}&=&\frac{1}{w_{t,n}(\bx)}\int_{\M}R_t(\bx,\by)T_tu_{t,n}(\by)p(\by)\mathd\by
+\frac{t}{w_{t,n}(\bx)}\int_\M\bar{R}(\bx,\by)u_{t,n}(\by)p(\by)\mathd\by,\nonumber\\
    T^2_tu_{t,n}&=&\frac{1}{w_{t,n}(\bx)}\int_{\M}R_t(\bx,\by)T_{t,n}u_{t,n}(\by)p(\by)\mathd\by
+\frac{t}{w_{t,n}(\bx)}\int_\M\bar{R}(\bx,\by)u_{t,n}(\by)p(\by)\mathd\by.\nonumber
  \end{eqnarray}
We will prove the theorem by upper bound $T_tu_{t,n}-T_t^1u_{t,n}$, $T_t^1u_{t,n}-T_t^2u_{t,n}$ and $T_t^2u_{t,n}-T_{t,n}u_{t,n}$ separately.

First, let us see $T_tu_{t,n}-T_t^1u_{t,n}$.
% Direct calculations give that
%   \begin{eqnarray}
%     \|w_t(\bx)-w_{t,n}(\bx)\|_\infty\le \sup_{g\in \mathcal{R}_t}(|p_n(g)-p(g)|)
%     ,\quad \|\nabla w_t(\bx)-\nabla w_{t,n}(\bx)\|_\infty\le \frac{1}{t^{1/2}}\sup_{g\in \mathcal{D}_t}(|p_n(g)-p(g)|).\nonumber
%   \end{eqnarray}
% and
%   \begin{eqnarray}
%     \left|\int_{\M}R_t(\bx,\by)T_tu_{t,n}(\by)\mathd\by\right|&\le& C t^{-k/4}\|T_tu_{t,n}\|_{L^2}\le Ct^{-k/4}\|u_{t,n}\|_{L^2}
% \le Ct^{-k/4}\|f\|_{\infty},\nonumber\\
% \left|\nabla_\bx\int_{\M}R_t(\bx,\by)T_tu_{t,n}(\by)\mathd\by\right|&\le& C t^{-(k+2)/4}\|T_tu_{t,n}\|_{L^2}\le Ct^{-(k+2)/4}\|f\|_{\infty},\nonumber
%   \end{eqnarray}
% and
%   \begin{eqnarray}
%     \left|\int_\M\bar{R}(\bx,\by)u_{t,n}(\by)\mathd\by\right|\le C\|u_{t,n}\|_{\infty}\le Ct^{-k/4}\|f\|_\infty,\nonumber
% \\  \left|\nabla_\bx \int_\M\bar{R}(\bx,\by)u_{t,n}(\by)\mathd\by\right|
% \le Ct^{-1/2}\|u_{t,n}\|_{\infty}\le Ct^{-(k+2)/4}\|f\|_\infty.\nonumber
%   \end{eqnarray}
% Using above inequalites, we have
  \begin{align*}
&    \left|T_tu_{t,n} - T^1_tu_{t,n}\right|\\
\le& \left|\frac{1}{w_{t,n}(\bx)}-\frac{1}{w_{t}(\bx)}\right|\left(
\left|\int_{\M}R_t(\bx,\by)T_tu_{t,n}(\by)p(\by)\mathd\by\right|
+t\left|\int_\M\bar{R}(\bx,\by)u_{t,n}(\by)p(\by)\mathd\by\right|\right)\\
\le& \frac{2C_t}{w_{\min}^2}\sup_{g\in \mathcal{R}_t}(|p_n(g)-p(g)|)\left(
\left|\int_{\M}R_t(\bx,\by)T_tu_{t,n}(\by)p(\by)\mathd\by\right|
+t\left|\int_\M\bar{R}(\bx,\by)u_{t,n}(\by)p(\by)\mathd\by\right|\right)\\
\le& \frac{C}{t^{3k/4}} (\|T_tu_{t,n}\|_{L^2}+t\|u_{t,n}\|_{L^2})\sup_{g\in \mathcal{R}_t}(|p_n(g)-p(g)|)\\
\le &\frac{C}{t^{3k/4}} \|u_{t,n}\|_{L^2}\sup_{g\in \mathcal{R}_t}(|p_n(g)-p(g)|)\\
 \le& \frac{C}{t^{3k/4}}\|f\|_\infty \sup_{g\in \mathcal{R}_t}(|p_n(g)-p(g)|),\nonumber\\
  \end{align*}
Similarly, we have
\begin{align*}
   \left|\nabla (T_tu_{t,n} - T^1_tu_{t,n})\right|\le&  \frac{C}{t^{(3k+2)/4}}\|f\|_\infty
 \sup_{g\in \mathcal{R}_t\cup\mathcal{D}_t}(|p_n(g)-p(g)|),\nonumber
\end{align*}
which proves that
  \begin{eqnarray}
\label{eqn:est-t1}
        \left\|T_tu_{t,n} - T^1_tu_{t,n}\right\|_{C^1}\le  \frac{C}{t^{(3k+2)/4}}\|f\|_\infty
\sup_{g\in \mathcal{R}_t\cup\mathcal{D}_t}(|p_n(g)-p(g)|).
  \end{eqnarray}
Secondly, using Theorem \ref{thm:dis_error} we have
  \begin{eqnarray}
&&    \left| T^1_tu_{t,n}- T^2_tu_{t,n}\right|\nonumber\\
&=&\left|\frac{1}{w_{t,n}(\bx)}
\int_{\M}R_t(\bx,\by)\left(T_tu_{t,n}(\by)-T_{t,n}u_{t,n}(\by)\right)p(\by)\mathd\by\right|\nonumber\\
&\le & Ct^{-k/4}\left\|T_tu_{t,n}-T_{t,n}u_{t,n}\right\|_{L^2}\nonumber\\
&= & Ct^{-k/4}\left\|(T_t-T_{t,n})T_{t,n}f\right\|_{L^2}\nonumber\\
&\le& \frac{C}{t^{3k/4+1}}\|f\|_\infty\left(
\sup_{g\in \mathcal{R}_t \cup \mathcal{R}_t\cdot \mathcal{K}_{t,n}}|p_n(g)-p(g)|+ t\sup_{g\in 
{\mathcal{K}}_{t,n}\cdot \overline{\mathcal{R}}_t}|p_n(g)-p(g)|+t^2\sup_{g\in \overline{\mathcal{K}}_{t,n}\cdot \overline{\mathcal{R}}_t}|p_n(g)-p(g)|\right).\nonumber
  \end{eqnarray}
and
 \begin{eqnarray}
&&    \left| \nabla \left(T^1_tu_{t,n}- T^2_tu_{t,n}\right)\right|\nonumber\\
&=&
\left|\nabla_\bx \left(\frac{1}{w_{t,n}(\bx)}\int_{\M}R_t(\bx,\by)\left(T_tu_{t,n}(\by)-T_{t,n}u_{t,n}(\by)\right)
p(\by)\mathd\by\right)\right|\nonumber\\
&\le & Ct^{-k/4+1/2}\left\|T_tu_{t,n}-T_{t,n}u_{t,n}\right\|_{L^2}\nonumber\\
&= & Ct^{-k/4+1/2}\left\|(T_t-T_{t,n})T_{t,n}f\right\|_{L^2}\nonumber\\
&\le& \frac{C}{t^{k/4+3/2}}\|f\|_\infty\left(
\sup_{g\in \mathcal{R}_t \cup \mathcal{R}_t\cdot \mathcal{K}_{t,n}}|p_n(g)-p(g)|+ t\sup_{g\in 
{\mathcal{K}}_{t,n}\cdot \overline{\mathcal{R}}_t}|p_n(g)-p(g)|+t^2\sup_{g\in \overline{\mathcal{K}}_{t,n}\cdot \overline{\mathcal{R}}_t}|p_n(g)-p(g)|\right).\nonumber
  \end{eqnarray}
This implies that
  \begin{eqnarray}
\label{eqn:est-t12}
    && \left\|T^1_tu_{t,n} - T^2_tu_{t,n}\right\|_{C^1}
\\
&\le&  \frac{C}{t^{k/4+3/2}}\|f\|_\infty\left(
\sup_{g\in \mathcal{R}_t \cup \mathcal{R}_t\cdot \mathcal{K}_{t,n}}|p_n(g)-p(g)|+ t\sup_{g\in 
{\mathcal{K}}_{t,n}\cdot \overline{\mathcal{R}}_t}|p_n(g)-p(g)|+t^2\sup_{g\in \overline{\mathcal{K}}_{t,n}\cdot \overline{\mathcal{R}}_t}|p_n(g)-p(g)|\right).\nonumber
  \end{eqnarray}
Now, we turn to estimate $T_{t,n}u_{t,n}-T^2_tu_{t,n}$. Using \eqref{eq:Ttn-1}, we have
  \begin{eqnarray}
    T_{t,n}u_{t,n}-T^2_tu_{t,n}&=&\frac{1}{w_{t,n}(\bx)}
\left(\frac{1}{n}\sum_{i=1}^nR_t(\bx,\bx_i)v_i-\int_{\M}R_t(\bx,\by)T_{t,n}u_{t,n}(\by)p(\by)\mathd\by\right)
\nonumber\\
&&
+\frac{t}{w_{t,n}(\bx)}\left(\frac{1}{n}\sum_{i=1}^n\bar{R}(\bx,\bx_i)u_i
-\int_\M\bar{R}(\bx,\by)u_{t,n}(\by)p(\by)\mathd\by\right).\nonumber
  \end{eqnarray}
  %  \begin{align*}
  %    T_{t,n}u_{t,n}=\frac{C_t}{n w_{t,n}(\bx)}\sum_{i=1}^n R\left(\frac{|\bx-\bx_i|^2}{4t}\right)u_i+\frac{tC_t }{n w_{t,n}(\bx)}\sum_{i=1}^n R\left(\frac{|\bx-\bx_i|^2}{4t}\right)f(\bx_i)
  %  \end{align*}
Using \eqref{eq:Ttn-1} again, the first term becomes
  \begin{eqnarray}
   && \left|\frac{1}{n}\sum_{i=1}^nR_t(\bx,\bx_i)v_i-\int_{\M}R_t(\bx,\by)T_{t,n}u_{t,n}(\by)p(\by)\mathd\by\right|
\nonumber\\
&\le& \left|\frac{1}{n}\sum_{i=1}^nR_t(\bx,\bx_i)\left(\frac{1}{n w_{t,n}(\bx_i)}\sum_{j=1}^n R_t\left(\bx_i,\bx_j\right)v_j
+\frac{t }{n w_{t,n}(\bx_i)}\sum_{j=1}^n \bar{R}_t\left(\bx_i-\bx_j\right)u_j\right)\right.\nonumber\\
&&\left.-\int_{\M}R_t(\bx,\by)
\left(\frac{1}{n w_{t,n}(\by)}\sum_{j=1}^n R_t\left(\by,\bx_j\right)v_j
+\frac{t }{n w_{t,n}(\by)}\sum_{j=1}^n \bar{R}_t\left(\by-\bx_j\right)u_j\right)p(\by)\mathd\by\right|\nonumber\\
&\le&  \left|\frac{1}{n}\sum_{j=1}^nv_j\left(\frac{1}{n}\sum_{i=1}^n\frac{R_t(\bx,\bx_i)}{ w_{t,n}(\bx_i)} R_t\left(\bx_i,\bx_j\right)
-\int_{\M}
\frac{R_t(\bx,\by)}{w_{t,n}(\by)} R_t\left(\by,\bx_j\right)p(\by)\mathd\by\right)\right|\nonumber\\
&&+\left|\frac{t}{n}\sum_{j=1}^nu_j\left(\frac{1}{n}\sum_{i=1}^n\frac{R_t(\bx,\bx_i)}{ w_{t,n}(\bx_i)} \bar{R}_t\left(\bx_i,\bx_j\right)
-\int_{\M}
\frac{R_t(\bx,\by)}{w_{t,n}(\by)} \bar{R}_t\left(\by,\bx_j\right)p(\by)\mathd\by\right)\right|\nonumber
% +\frac{t }{n w_{t,n}(\bx_i)}\sum_{j=1}^n R_t\left(\bx_i-\bx_j\right)u_j\right)\right.\nonumber\\
% &&\left.-\int_{\M}R_t(\bx,\by)
% \left(\frac{1}{n w_{t,n}(\by)}\sum_{j=1}^n R_t\left(\by,\bx_j\right)v_j
% +\frac{t }{n w_{t,n}(\by)}\sum_{j=1}^n R_t\left(\by-\bx_j\right)u_j\right)\mathd\by\right|\nonumber\\
% &\le& \left(\frac{1}{n}\sum_{j=1}^nv_j^2\right)^{1/2} \sup_{g\in \mathcal{K}_{t,n}\cdot \mathcal{R}_t}|p_n(g)-p(g)|
% +t\left(\frac{1}{n}\sum_{j=1}^nu_j^2\right)^{1/2}\sup_{g\in \mathcal{K}_{t,n}\cdot \overline{\mathcal{R}}_t}|p_n(g)-p(g)|\nonumber\\
% &\le & C\|f\|_\infty\left(\sup_{g\in \mathcal{K}_{t,n}\cdot \mathcal{R}_t}|p_n(g)-p(g)|
% +t\sup_{g\in \mathcal{K}_{t,n}\cdot \overline{\mathcal{R}}_t}|p_n(g)-p(g)|\right)
  \end{eqnarray}
Using the similar derivation from \eqref{eq:error-b-1-sum} to \eqref{eq:error-b}, we can get
  \begin{eqnarray}
   && \left|\frac{1}{n}\sum_{i=1}^nR_t(\bx,\bx_i)v_i-\int_{\M}R_t(\bx,\by)T_{t,n}u_{t,n}(\by)p(\by)\mathd\by\right|
\nonumber\\
&\le&\frac{C}{t^{k/4}} \left(\frac{1}{n}\sum_{j=1}^nv_j^2\right)^{1/2} C_t\sup_{g\in \mathcal{K}_{t,n}\cdot \mathcal{R}_t}|p_n(g)-p(g)|
+\frac{C}{t^{k/4-1}}\left(\frac{1}{n}\sum_{j=1}^nu_j^2\right)^{1/2}C_t\sup_{g\in \mathcal{K}_{t,n}\cdot \overline{\mathcal{R}}_t}|p_n(g)-p(g)|\nonumber\\
&\le & \frac{C}{t^{3k/4}}\|f\|_\infty\left(\sup_{g\in \mathcal{K}_{t,n}\cdot \mathcal{R}_t}|p_n(g)-p(g)|
+t\sup_{g\in \mathcal{K}_{t,n}\cdot \overline{\mathcal{R}}_t}|p_n(g)-p(g)|\right)\nonumber
  \end{eqnarray}
The second term can be bounded similarly,
  \begin{eqnarray}
   && \left|\frac{1}{n}\sum_{i=1}^n\bar{R}(\bx,\bx_i)u_i
-\int_\M\bar{R}(\bx,\by)u_{t,n}(\by)p(\by)\mathd\by\right|\nonumber\\
&\le &\frac{C}{t^{k/4}}\left(\frac{1}{n}\sum_{j=1}^nu_j^2\right)^{1/2} C_t\sup_{g\in \mathcal{K}_{t,n}\cdot \overline{\mathcal{R}}_t}|p_n(g)-p(g)|
+\frac{C}{t^{k/4-1}}\left(\frac{1}{n}\sum_{j=1}^nf_j^2\right)^{1/2}C_t\sup_{g\in \overline{\mathcal{K}}_{t,n}\cdot \overline{\mathcal{R}}_t}|p_n(g)-p(g)|\nonumber\\
&\le & \frac{C}{t^{3k/4}}\|f\|_\infty\left(\sup_{g\in \mathcal{K}_{t,n}\cdot \overline{\mathcal{R}}_t}|p_n(g)-p(g)|
+t\sup_{g\in \overline{\mathcal{K}}_{t,n}\cdot \overline{\mathcal{R}}_t}|p_n(g)-p(g)|\right)
  \end{eqnarray}
Now, we have
  \begin{eqnarray}
    |T_{t,n}u_{t,n}-T^2_tu_{t,n}|\le \frac{C}{t^{3k/4}}\|f\|_\infty\left(\sup_{g\in \mathcal{K}_{t,n}\cdot \mathcal{R}_t}|p_n(g)-p(g)|
+t\sup_{g\in \mathcal{K}_{t,n}\cdot \overline{\mathcal{R}}_t}|p_n(g)-p(g)|
+t^2\sup_{g\in \overline{\mathcal{K}}_{t,n}\cdot \overline{\mathcal{R}}_t}|p_n(g)-p(g)|\right)\nonumber
  \end{eqnarray}
Using the similar method, we can get
 \begin{eqnarray}
    |\nabla(T_{t,n}u_{t,n}-T^2_tu_{t,n})|\le \frac{C}{t^{3k/4+1/2}}\|f\|_\infty\left(\sup_{g\in \mathcal{K}_{t,n}\cdot \mathcal{D}_t}|p_n(g)-p(g)|
+t\sup_{g\in \mathcal{K}_{t,n}\cdot \overline{\mathcal{D}}_t}|p_n(g)-p(g)|
+t^2\sup_{g\in \overline{\mathcal{K}}_{t,n}\cdot \overline{\mathcal{D}}_t}|p_n(g)-p(g)|\right)\nonumber
  \end{eqnarray}
The estimate of $\|(T_t-T_{t,n})T_{t,n}\|_{C^1}$ in Theorem \ref{thm:converge_c1} is proved.

Similarly, we can obtain the estimate of  $\|(T_t-T_{t,n})f\|_{C^1}$ for any $f\in C(\M)$ which complete the proof.

% Using Theorem \ref{thm:bound}, we have
%   \begin{eqnarray}
%     \left|T_{t,n}u\right|&\le&  Ct^{-k/4}\|u\|_{\infty},\\
% \left|\nabla T_{t,n}u\right|&\le&  Ct^{-(k+2)/4}\|u\|_{\infty}.
%   \end{eqnarray}
% Thus,
%   \begin{eqnarray}
%     \left|T_{t,n}u-T^2_tu\right|&\le& \frac{Ch}{t^{(k+2)/4}}\|u\|_{\infty},
%  \\
%     \left|\nabla\left(T_{t,n}u-T^2_tu\right)\right|&\le& \frac{Ch}{t^{k/4+1}}\|u\|_{\infty}.
%   \end{eqnarray}
% which also reads
%   \begin{eqnarray}
% \label{eqn:est-t2}
%     \left\|T_{t,n}u-T^2_tu\right\|_{C^1}\le \frac{Ch}{t^{k/4+1}}\|u\|_{\infty}
%   \end{eqnarray}
%The proof is completed.% by combining \eqref{eqn:est-t1}, \eqref{eqn:est-t12} and \eqref{eqn:est-t2}.
  % \begin{eqnarray}
  %   \left\|T_{t,n}u-T_tu\right\|_{C^1}\le \frac{Ch}{t^{k/4+2}}\|u\|_{C^1}
  % \end{eqnarray}
  % \begin{eqnarray}
  %   \left\|T_{t,n}-T_t\right\|_{C^1}\le \frac{Ch}{t^{k/4+2}}
  % \end{eqnarray}
\end{proof}

\section{Conclusions}

In this paper, we proved that the spectra of the normalized graph laplacian \eqref{eqn:eigen_dis} will converge to the spectral of a weighted Laplace-Beltrami operator
with Neumann boundary condition \eqref{eqn:eigen_neumann} as $t\rightarrow 0$ and the number of sample points goes to infinity. The samples points are assumed to be 
drawn on a smooth manifold according to some probability distribution $p$. Moreover, we also give an estimate of the convergence rate. Up to our knowledge, this is the 
first result about the spectra convergence rate of graph laplacian. However, the estimate of the convergence rate in this paper is far from optimal. There are mainly two
places in the analysis which can be improved in the future. The first one is the estimate of the integral equation \eqref{eq:integral}. Now, we only get $L^2$ estimate, 
however, in the spectra convergence analysis, we need $C^1$ estimate. In this paper, the regularity is lifted by using the regularity of the kernel function. The trade off 
is that a large number $t^{-k/4}$ emerge which reduce the rate of convergence. The other place is the estimate of the covering number. The estimate of the covering number 
is very rough in this paper. More delicate method would give better estimate which could help to improve the estimate of the convergence rate. 

\appendix
\setcounter{section}{1}
\setcounter{equation}{0}
\vspace{5mm}
\noindent
\textbf{Appendix A: Proof of Theorem \ref{thm:elliptic_dis}}

\begin{proposition}(\cite{SS-neumann})
Assume both $\M$ and $\p \M$ are $C^2$ smooth.
There are constants $w_{\min}>0, w_{\max}<+\infty$ and $T_0>0$ depending only on the geometry of $\M$,
so that $$w_{\min} \leq w_t(\bx)=\int_\M R_t(\bx, \by) \mathd \by \leq w_{\max} $$
as long as $t<T_0$.
\label{prop:bound_int_R_t}
\end{proposition}

We have the following lemma about the function $w_{t, n}$.
\begin{lemma}
Under the assumptions in Assumption \ref{assumptions},
if $\D C_t\sup_{f\in \mathcal{R}_t}|p(f)-p_n(f)|\le w_{\min}/2$,
$$w_{\min}/2\leq w_{t, n}(\bx) \leq w_{\max}+w_{\min}/2. $$
\label{lem:bound_w_t_h}
\end{lemma}
This lemma is a direct consequence of Proposition \ref{prop:bound_int_R_t} and the fact that
$$\left|w_{t,n}(\bx) - C_t\int_\M R\left(\frac{|\bx-\by|^2}{4t}\right) p(\by)\mathd \by\right| \leq C_t\sup_{f\in \mathcal{R}_t}|p(f)-p_n(f)|.$$

\begin{lemma}(\cite{SS-neumann,SS-iso})
\label{lem:coercivity}
 For any function $u\in L^2(\mathcal{M})$,
there exists a constant $C>0$ only depends on $\M$, such that
  \begin{eqnarray}
    \int_{\mathcal{M}}\int_{\mathcal{M}} R_t(\bx,\by)(u(\bx)-u(\by))^2p(\bx)p(\by)\mathd\bx\mathd\by \ge C\int_\mathcal{M} |u(\bx)-\bar{u}|^2p(\bx)\mathd \bx,
  \end{eqnarray}
where
$$\bar{u}=\int_\M u(\bx)p(\bx)\mathd\bx.$$
% where
% \begin{eqnarray}
% v(\bx)=\frac{1}{w_t(\bx)}\int_{\mathcal{M}}R_t(\bx,\by)u(\by)p(\by)\mathd \by,
% \end{eqnarray}
% and $\D w_t(\bx) = \int_{\mathcal{M}}R_t\left(\bx,\by\right)p(\by)\mathd \by$.
\end{lemma}
Now, we can prove Theorem \ref{thm:elliptic_dis}.
\begin{proof} {\it of Theorem \ref{thm:elliptic_dis}}

First, we introduce a smooth function $u$ that approximates $\bfu$ at the samples $X_n$.
\begin{eqnarray}
  u(\bx)=\frac{C_t}{nw_{t',n}(\bx)}\sum_{i=1}^nR\left(\frac{|\bx-\bfp_i|^2}{4t'}\right)u_i,
\label{eqn:def_dis_u}
\end{eqnarray}
where $w_{t',n}(\bx)=\frac{C_t}{n}\sum_{i=1}^nR\left(\frac{|\bx-\bfp_i|^2}{4t'}\right)$ and $t'=t/18$. 

Then, we have
%Since $R$ is compactly supported, consider
%\begin{eqnarray}
    %\mathcal{M}_{\bx}^t=\left\{\by\in \mathcal{M}: |\by-\bx|^2\le 4t\right\}.
%\end{eqnarray}
\begin{eqnarray}
&&\int_{\mathcal{M}}  \int_{\mathcal{M}}R_{t'}(\bx,\by) \left(u(\bx)-u(\by)\right)^2p(\bx)p(\by)\mathd \bx \mathd\by\nonumber\\
&=& \int_{\mathcal{M}}  \int_{\mathcal{M}}R_{t'}(\bx,\by) \left(\frac{1}{nw_{t',n}(\bx)}\sum_{i=1}^nR_{t'}(\bx,\bfp_i)u_i
-\frac{1}{nw_{t',n}(\by)}\sum_{j=1}^nR_{t'}(\bfp_j,\by)u_j\right)^2p(\bx)p(\by)\mathd \bx \mathd\by\nonumber\\
&=& \int_{\mathcal{M}}  \int_{\mathcal{M}}R_{t'}(\bx,\by) \left(\frac{1}{n^2w_{t',n}(\bx)w_{t',n}(\by)}\sum_{i,j=1}^nR_{t'}(\bx,\bfp_i)R_{t'}(\bfp_j,\by)
(u_i-u_j)\right)^2p(\bx)p(\by)\mathd \bx \mathd\by\nonumber\\
&\le & \int_{\mathcal{M}}  \int_{\mathcal{M}}R_{t'}(\bx,\by) \frac{1}{n^2w_{t',n}(\bx)w_{t',n}(\by)}\sum_{i,j=1}^nR_{t'}(\bx,\bfp_i)R_{t'}(\bfp_j,\by)
(u_i-u_j)^2p(\bx)p(\by)\mathd \bx \mathd\by\nonumber\\
&=&  \frac{1}{n^2}\sum_{i,j=1}^n\left(\int_{\mathcal{M}}  \int_{\mathcal{M}}\frac{1}{w_{t',n}(\bx)w_{t',n}(\by)}
R_{t'}(\bx,\bfp_i)R_{t'}(\bfp_j,\by)R_{t'}(\bx,\by)p(\bx)p(\by)\mathd \bx \mathd\by\right)(u_i-u_j)^2.\nonumber\\
%&=&  \sum_{i,j=1}^n\left(\int_{\mathcal{M}_{\bfp_j}^{t'}}  \int_{\mathcal{M}_{\bfp_i}^{t'}}\frac{1}{w_{t',n}(\bx)w_{t',n}(\by)}
%R_{t'}(\bx,\bfp_i)R_{t'}(\bfp_j,\by)R_{t'}(\bx,\by)\mathd \bx \mathd\by\right)V_i
%V_j(u_i-u_j)^2
\label{eqn:A1}
\end{eqnarray}
Denote
%$$A = \int_{\mathcal{M}_{\bfp_j}^{t'}}  \int_{\mathcal{M}_{\bfp_i}^{t'}}\frac{1}{w_{t',n}(\bx)w_{t',n}(\by)}
%R_{t'}(\bx,\bfp_i)R_{t'}(\bfp_j,\by)R_{t'}(\bx,\by)\mathd \bx \mathd\by$$
$$A = \int_{\mathcal{M}}  \int_{\mathcal{M}}\frac{1}{w_{t',n}(\bx)w_{t',n}(\by)}
R_{t'}(\bx,\bfp_i)R_{t'}(\bfp_j,\by)R_{t'}(\bx,\by)p(\bx)p(\by)\mathd \bx \mathd\by$$
and then notice only when $|\bfp_i-\bfp_j|^2\le 36t'$ is $A \neq 0$. For
$|\bfp_i-\bfp_j|^2\le 36t'$, we have
\begin{eqnarray}
A &\le&\int_{\mathcal{M}}   \int_{\mathcal{M}}
R_{t'}(\bx,\bfp_i)R_{t'}(\bfp_j,\by)R_{t'}(\bx,\by)R\left(\frac{|\bfp_i-\bfp_j|^2}{72t'}\right)^{-1}R\left(\frac{|\bfp_i-\bfp_j|^2}{72t'}\right)
p(\bx)p(\by)\mathd \bx \mathd\by\nonumber\\
&\le&\frac{CC_t}{\delta_0}  \int_{\mathcal{M}}   \int_{\mathcal{M}}
R_{t'}(\bx,\bfp_i)R_{t'}(\bfp_j,\by)R\left(\frac{|\bfp_i-\bfp_j|^2}{72t'}\right)p(\bx)p(\by)\mathd \bx \mathd\by\nonumber\\
&\le&CC_t  \int_{\mathcal{M}}   \int_{\mathcal{M} }
R_{t'}(\bx,\bfp_i)R_{t'}(\bfp_j,\by)R\left(\frac{|\bfp_i-\bfp_j|^2}{72t'}\right)p(\bx)p(\by)\mathd \bx \mathd\by\nonumber\\
&\le & CC_t R\left(\frac{|\bfp_i-\bfp_j|^2}{4t}\right).
\label{eqn:A2}
\end{eqnarray}
%\begin{eqnarray}
%A &\le&\int_{\mathcal{M}_{\bfp_j}^{t'}}   \int_{\mathcal{M}_{\bfp_i}^{t'}\cap \mathcal{M}_{\by}^{t'} }
%R_{t'}(\bx,\bfp_i)R_{t'}(\bfp_j,\by)R_{t'}(\bx,\by)R\left(\frac{|\bfp_i-\bfp_j|^2}{72t'}\right)^{-1}R\left(\frac{|\bfp_i-\bfp_j|^2}{72t'}\right)\mathd \bx \mathd\by\nonumber\\
%&\le&\frac{CC_t}{\delta_0}  \int_{\mathcal{M}_{\bfp_j}^{t'}}   \int_{\mathcal{M}_{\bfp_i}^{t'}\cap \mathcal{M}_{\by}^{t'} }
%R_{t'}(\bx,\bfp_i)R_{t'}(\bfp_j,\by)R\left(\frac{|\bfp_i-\bfp_j|^2}{72t'}\right)\mathd \bx \mathd\by\nonumber\\
%&\le&CC_t  \int_{\mathcal{M}}   \int_{\mathcal{M} }
%R_{t'}(\bx,\bfp_i)R_{t'}(\bfp_j,\by)R\left(\frac{|\bfp_i-\bfp_j|^2}{72t'}\right)\mathd \bx \mathd\by\nonumber\\
%&\le & CC_t R\left(\frac{|\bfp_i-\bfp_j|^2}{4t}\right)
%\end{eqnarray}

%\begin{eqnarray}
%&&\int_{\mathcal{M}}  \int_{\mathcal{M}}R_{t'}(\bx,\by) \left(u(\bx)-u(\by)\right)^2\mathd \bx \mathd\by\nonumber\\
%&\le& C \sum_{i,j=1}^nR\left(\frac{|\bfp_i-\bfp_j|^2}{4t}\right)(u_i-u_j)^2A_iA_j
%\end{eqnarray}
Combining Equation~\eqref{eqn:A1}, \eqref{eqn:A2} and Lemma~\ref{lem:coercivity}, we obtain
\begin{eqnarray}
\frac{CC_t}{n^2t}  \sum_{i,j=1}^nR\left(\frac{|\bfp_i-\bfp_j|^2}{4t}\right)(u_i-u_j)^2\ge \int_{\mathcal{M}}  (u(\bx)-\bar{u})^2p(\bx)\mathd\bx
\label{eqn:u_2V_0}
\end{eqnarray}

We now lower bound the RHS of the above equation using $\frac{1}{n}\sum_{j=1}^n u_i^2$.
\begin{eqnarray}
 |\bar{u}|&=&\left|\int_{\mathcal{M}}u(\bx)p(\bx)\mathd \bx\right|
=\left| \frac{1}{n}\sum_{j=1}^n \left( u_j \int_{\mathcal{M}}  \frac{C_t}{w_{t',n}(\bx)}R\left(\frac{|\bx-\bfp_j|^2}{4t'}\right) p(\bx)
\mathd \bx \right) \right|.
\end{eqnarray}
% Let  $q(\bx) = \frac{C_t}{w_{t',n}(\bx)}{R\left(\frac{|\bx-\bfp_j|^2}{4t'}\right)}$. There exists a constant $C$ so that $|q(\bx)|\leq CC_t $ and
% \begin{eqnarray*}
% |\nabla q(\bx)| \leq \frac{C_t}{w_{t',n}(\bx)} \left|\nabla R\left(\frac{|\bx-\bfp_j|^2}{4t'}\right)\right| + \frac{C_t \left|\nabla w_{t',n}(\bx)\right|}{w^2_{t,h}(\bx)} R\left(\frac{|\bx-\bfp_j|^2}{4t'}\right) \leq \frac{CC_t}{t^{1/2}}.
% \end{eqnarray*}
%Then, for sufficiently small $t, \frac{h}{t^{1/2}}$, there exists a constant $C$
Notice that
\begin{eqnarray}
\left| \int_{\mathcal{M}}  \frac{C_t}{w_{t',n}(\bx)}R\left(\frac{|\bx-\bfp_j|^2}{4t'}\right) p(\bx)\mathd \bx  -
\frac{1}{n}\sum_{i=1}^n \frac{C_t}{w_{t',n}(\bfp_i)}R\left(\frac{|\bfp_i-\bfp_j|^2}{4t'}\right)\right| \leq C_t\sup_{f\in \mathcal{K}_{t',n}}|p(f)-p_n(f)|.\nonumber
\end{eqnarray}
Thus we have
\begin{eqnarray}			
|\bar{u}| &\le& \left| \frac{1}{n^2}\sum_{i,j=1}^n \frac{C_t}{w_{t',n}(\bfp_i)}R\left(\frac{|\bfp_i-\bfp_j|^2}{4t'}\right)u_j \right|
+  \left(\frac{1}{n}\sum_{j=1}^n |u_j|\right) \sup_{f\in \mathcal{K}_{t',n}}|p(f)-p_n(f)| \nonumber\\
&\le& \left|\frac{1}{n} \sum_{i=1}^n u(\bfp_i)\right| + \left(\frac{1}{n}\sum_{j=1}^n |u_j|\right) \sup_{f\in \mathcal{K}_{t',n}}|p(f)-p_n(f)| \nonumber \\
&\le&\left| \frac{1}{n^2}\sum_{i,j=1}^n\frac{C_t}{w_{t',n}(\bfp_i)}R\left(\frac{|\bfp_i-\bfp_j|^2}{4t'}\right)(u_j-u_i)\right|
+  \left(\frac{1}{n}\sum_{j=1}^n u_j^2\right)^{1/2} \sup_{f\in \mathcal{K}_{t',n}}|p(f)-p_n(f)|\nonumber\\
&\le & \frac{2}{w_{\min}}\left(\frac{C_t}{n^2}
\sum_{i,j=1}^nR\left(\frac{|\bfp_i-\bfp_j|^2}{4t'}\right)(u_i-u_j)^2\right)^{1/2}+\left(\frac{1}{n}\sum_{j=1}^n u_j^2\right)^{1/2} \sup_{f\in \mathcal{K}_{t',n}}|p(f)-p_n(f)|,\nonumber\\
\label{eqn:u_2V_1}
\end{eqnarray}
Denote
\begin{eqnarray}
A = \int_{\mathcal{M}}\frac{C_t}{w_{t',n}^2(\bx)}R\left(\frac{|\bx-\bfp_i|^2}{4t'}\right)
R\left(\frac{|\bx-\bfp_l|^2}{4t'}\right)p(\bx)\mathd\bx-\nonumber\\
\frac{1}{n}\sum_{j=1}^n\frac{C_t}{w_{t',n}^2(\bfp_j)}R\left(\frac{|\bfp_j-\bfp_i|^2}{4t'}\right)R\left(\frac{|\bfp_j-\bfp_l|^2}{4t'}\right)\nonumber
\end{eqnarray}
and then $\D|A|\le C_t\sup_{f\in \mathcal{K}_{t',n}\cdot\mathcal{K}_{t',n}}|p(f)-p_n(f)|$. At the same time, notice that only when $|\bfp_i-\bfp_l|^2 <16t'$ is $A\neq 0$. Thus we have
\begin{eqnarray}
|A| \le \frac{1}{\delta_0} |A| R(\frac{|\bfp_i-\bfp_l|^2}{72t'}).\nonumber
\end{eqnarray}
Then
\begin{eqnarray}
&&\left|\int_{\mathcal{M}}u^2(\bx)\mathd\bx-\frac{1}{n}\sum_{j=1}^nu^2(\bfp_j)\right| \nonumber\\
&\le&
  \frac{1}{n^2}\sum_{i,l=1}^n|C_tu_iu_l| |A| \nonumber\\
&\le &\frac{C_t}{n^2}\sup_{f\in \mathcal{K}_{t',n}\cdot\mathcal{K}_{t',n}}|p(f)-p_n(f)| \sum_{i, l=1}^n\left|C_t R\left(\frac{|\bfp_i-\bfp_l|^2}{72t'}\right) u_iu_l \right|\nonumber\\
&\le &\frac{C_t}{n^2}\sup_{f\in \mathcal{K}_{t',n}\cdot\mathcal{K}_{t',n}}|p(f)-p_n(f)|\sum_{i, l=1}^n C_t R\left(\frac{|\bfp_i-\bfp_l|^2}{72t'}\right) u_i^2\nonumber\\
& \le&
(w_{\max}+w_{\min}/2)C_t\sup_{f\in \mathcal{K}_{t',n}\cdot\mathcal{K}_{t',n}}|p(f)-p_n(f)| \left(\frac{1}{n}\sum_{i=1}^n u_i^2\right).
\label{eqn:u_2V_2}
\end{eqnarray}
In the last inequality, we use the condition that $C_t\sup_{f\in \mathcal{R}_{t}}|p(f)-p_n(f)|\le w_{\min}/2 $.

Now combining Equation~\eqref{eqn:u_2V_0},~\eqref{eqn:u_2V_1} and \eqref{eqn:u_2V_2}, we have for small $t$
\begin{eqnarray}
  \frac{1}{n}\sum_{i=1}^nu^2(\bfp_i) &=& \int_{\mathcal{M}}  u^2(\bx)p(\bx)\mathd\bx
+(w_{\max}+w_{\min}/2)C_t\sup_{f\in \mathcal{K}_{t',n}\cdot\mathcal{K}_{t',n}}|p(f)-p_n(f)| \left(\frac{1}{n}\sum_{i=1}^n u_i^2\right) \nonumber \\
&\le& 2\int_{\mathcal{M}}  (u(\bx)-\bar{u})^2p(\bx)\mathd\bx+2\bar{u}^2+
(w_{\max}+w_{\min}/2)C_t\sup_{f\in \mathcal{K}_{t',n}\cdot\mathcal{K}_{t',n}}|p(f)-p_n(f)| \left(\frac{1}{n}\sum_{i=1}^n u_i^2\right)
\nonumber\\
&\le & \frac{CC_t}{n^2t}  \sum_{i,j=1}^nR\left(\frac{|\bfp_i-\bfp_j|^2}{4t}\right)(u_i-u_j)^2\nonumber\\
&&
+\max\{w_{\max}+w_{\min}/2,2/w_{\min}\}C_t\sup_{f\in \mathcal{K}_{t',n}\cdot\mathcal{K}_{t',n}\cup \mathcal{K}_{t',n}}|p(f)-p_n(f)| \left(\frac{1}{n}\sum_{i=1}^n u_i^2\right).\nonumber
\end{eqnarray}
% Here we use the fact that for $t=18t'$.

%  $$R\left(\frac{|\bfp_i-\bfp_j|^2}{4t'}\right) \leq \frac{1}{\delta_0} R\left(\frac{|\bfp_i-\bfp_j|^2}{4t}\right).$$

%\begin{eqnarray}
 %&& \left|\int_{\mathcal{M}}u^2(\bx)\mathd\bx-\sum_{j=1}^nu^2(\bfp_j)V_j\right|\nonumber\\
%&= & \left|\sum_{i,l=1}^nC_t^2u_iu_lV_iV_l\left(\int_{\mathcal{M}}\frac{1}{w_{t',n}^2(\bx)}R\left(\frac{|\bx-\bfp_i|^2}{4t'}\right)
%R\left(\frac{|\bx-\bfp_l|^2}{4t'}\right)\mathd\bx-\right.\right.\nonumber\\
%&&\left.\left.\sum_{j=1}^n\frac{1}{w_{t',n}^2(\bfp_j)}R\left(\frac{|\bfp_j-\bfp_i|^2}{4t'}\right)
%R\left(\frac{|\bfp_j-\bfp_l|^2}{4t'}\right)\right)\right|\nonumber\\
%&\le &\frac{Ch}{t^{1/2}}\left|\sum_{i,l=1}^nC_t^2u_iu_lV_iV_l\chi_{\Omega_i^t}(l)\right|\nonumber\\
%&\le & \frac{Ch}{t^{1/2}}\left|\sum_{i=1}^nC_tu_iV_i\left(\sum_{l} V_l\chi_{\Omega_i^t}(l)\right)^{1/2}\left(\sum_{l}u_l^2V_l\chi_{\Omega_i^t}(l)\right)^{1/2}\right|\nonumber\\
%&\le & \frac{Ch}{t^{1/2}}\left|\sqrt{C_t}\sum_{i=1}^nu_iV_i\left(\sum_{l}u_l^2V_l\chi_{\Omega_i^t}(l)\right)^{1/2}\right|\nonumber\\
%&\le & \frac{Ch}{t^{1/2}}\left|\sqrt{C_t}\left(\sum_{i=1}^nu_i^2V_i\right)^{1/2}\left(\sum_{i=1}^nV_i\sum_{l}u_l^2V_l\chi_{\Omega_i^t}(l)\right)^{1/2}\right|\nonumber\\
%&\le & \frac{Ch}{t^{1/2}}\left|\sqrt{C_t}\left(\sum_{i=1}^nu_i^2V_i\right)^{1/2}\left(\sum_lu_l^2V_l\sum_{i=1}^nV_i\chi_{\Omega_l^t}(i)\right)^{1/2}\right|\nonumber\\
%&\le & \frac{Ch}{t^{1/2}}\left(\sum_{i=1}^nu_i^2V_i\right)
%\end{eqnarray}

Let $\delta=\frac{w_{\min}}{4w_{\max}+3w_{\min}}$. If $\frac{1}{n}\sum_{i=1}^nu^2(\bfp_i)\ge \frac{\delta^2}{n} \sum_{i=1}^nu_i^2$, and 
\begin{align*}
  \max\{w_{\max}+w_{\min}/2,2/w_{\min}\}C_t\sup_{f\in \mathcal{K}_{t',n}\cdot\mathcal{K}_{t',n}\cup \mathcal{K}_{t',n}}|p(f)-p_n(f)|\le \delta^2/2
\end{align*}
then we have completed the proof.
Otherwise, we have
\begin{eqnarray}
\frac{1}{n}\sum_{i=1}^n( u_i-u(\bfp_i))^2 = \frac{1}{n}\sum_{i=1}^n u_i^2 + \frac{1}{n}\sum_{i=1}^n u(\bfp_i)^2 -
\frac{2}{n} \sum_{i=1}^n u_iu(\bfp_i) \ge \frac{(1- \delta)^2}{n}\sum_{i=1}^n u_i^2.
\end{eqnarray}
This enables us to prove the theorem in the case of $\frac{1}{n}\sum_{i=1}^nu^2(\bfp_i) < \frac{\delta^2}{n} \sum_{i=1}^nu_i^2$ as follows.
\begin{eqnarray}
&&  \frac{C_t}{n^2}\sum_{i,j=1}^nR\left(\frac{|\bfp_i-\bfp_j|^2}{4t'}\right)(u_i-u_j)^2\nonumber\\
&=&   \frac{2C_t}{n^2}\sum_{i,j=1}^nR\left(\frac{|\bfp_i-\bfp_j|^2}{4t'}\right)u_i(u_i-u_j)\nonumber\\
&=& \frac{2}{n}\sum_{i=1}^nu_i(u_i-u(\bfp_i))w_{t',n}(\bfp_i)\nonumber\\
&=& \frac{2}{n}\sum_{i=1}^n(u_i-u(\bfp_i))^2w_{t',n}(\bfp_i)+\frac{2}{n}\sum_{i=1}^nu(\bfp_i)(u_i-u(\bfp_i))w_{t',n}(\bfp_i)\nonumber\\
&\ge& \frac{2}{n}\sum_{i=1}^n(u_i-u(\bfp_i))^2w_{t',n}(\bfp_i)-2\left(\frac{1}{n}\sum_{i=1}^nu^2(\bfp_i)w_{t',n}(\bfp_i)\right)^{1/2}
\left(\frac{1}{n}\sum_{i=1}^n(u_i-u(\bfp_i))^2w_{t,n}(\bfp_i)\right)^{1/2}\nonumber\\
&\ge& \frac{w_{\min}}{n}\sum_{i=1}^n(u_i-u(\bfp_i))^2-2(w_{\max}+w_{\min}/2)
\left(\frac{1}{n}\sum_{i=1}^nu^2(\bfp_i)\right)^{1/2}\left(\frac{1}{n}\sum_{i=1}^n(u_i-u(\bfp_i))^2\right)^{1/2}\nonumber\\
&\ge& (w_{\min}(1-\delta)-2(w_{\max}+w_{\min}/2)\delta)\left(\frac{1}{n}\sum_{i=1}^nu_i^2\right)^{1/2}\left(\frac{1}{n}\sum_{i=1}^n(u_i-u(\bfp_i))^2\right)^{1/2}\nonumber\\
&\ge& w_{\min}(1-\delta)^2\left(\frac{1}{n}\sum_{i=1}^nu_i^2\right).
\end{eqnarray}
%The theorem can be proved using the fact that
%\begin{equation*}
  %\sum_{i,j=1}^nR\left(\frac{|\bfp_i-\bfp_j|^2}{4t'}\right)(u_i-u_j)^2V_iV_j\le \frac{1}{\delta_0}\sum_{i,j=1}^nR\left(\frac{|\bfp_i-\bfp_j|^2}{4t}\right)(u_i-u_j)^2V_iV_j.
%\end{equation*}

%\begin{eqnarray}
%&&\left|\int_{\mathcal{M}}u(\bx)\mathd \bx-\sum_{j=1}^nu(\bfp_j)V_j\right|\nonumber\\
%&= & \left|\sum_{i=1}^nu_iV_i \left(\int_{\mathcal{M}}\frac{1}{w_{t,n}(\bx)}R\left(\frac{|\bfp_i-\bx|^2}{4t}\right)\mathd\bx-\sum_{j=1}^n
%\frac{1}{w_{t,n}(\bfp_j)}R\left(\frac{|\bfp_i-\bfp_j|^2}{4t}\right)V_j\right)\right|\nonumber\\
%&\le & C\frac{h}{t^{1/2}}\sum_{i=1}^n|u_i|V_i\le C\frac{h}{t^{1/2}}\left(\sum_{i=1}^nV_i\right)^{1/2}\left(\sum_{i=1}^nu_i^2V_i\right)^{1/2}=C\frac{h}{t^{1/2}}\left(\sum_{i=1}^nu_i^2V_i\right)^{1/2}
%\end{eqnarray}
\end{proof}

\bibliographystyle{abbrv}
\bibliography{poisson}

\begin{thebibliography}{10}

\bibitem{Atkinson67}
K.~Atkinson.
\newblock The numerical solution of the eigenvalue problem for compact integral
  operators.
\newblock {\em Transactions of the American Mathematical Society},
  129(3):458--465, 1967.

\bibitem{belkin2003led}
M.~Belkin and P.~Niyogi.
\newblock Laplacian eigenmaps for dimensionality reduction and data
  representation.
\newblock {\em Neural Computation}, 15(6):1373--1396, 2003.

\bibitem{CLEM_08}
M.~Belkin and P.~Niyogi.
\newblock Convergence of laplacian eigenmaps.
\newblock In {\em Adv. Neur. In.: Proceedings of the 2006 Conference},
  volume~19, page 129, 2007.

\bibitem{BelkinQWZ12}
M.~Belkin, Q.~Que, Y.~Wang, and X.~Zhou.
\newblock Toward understanding complex spaces: Graph laplacians on manifolds
  with singularities and boundaries.
\newblock In {\em 25th Annual Conference on Learning Theory}, volume~23, pages
  36.1--36.24, 2012.

\bibitem{Chung}
F.~R.~K. Chung.
\newblock {\em Spectral Graph Theory}.
\newblock American Mathematical Society, 1997.

\bibitem{Coifman05geometricdiffusions}
R.~R. Coifman, S.~Lafon, A.~B. Lee, M.~Maggioni, F.~Warner, and S.~Zucker.
\newblock Geometric diffusions as a tool for harmonic analysis and structure
  definition of data: Diffusion maps.
\newblock In {\em Proceedings of the National Academy of Sciences}, pages
  7426--7431, 2005.

\bibitem{Lafon04diffusion}
S.~Lafon.
\newblock {\em Diffusion Maps and Geodesic Harmonics}.
\newblock PhD thesis, 2004.

\bibitem{LSS}
Z.~Li, Z.~Shi, and J.~Sun.
\newblock Point integral method for solving poisson-type equations on manifolds
  from point clouds with convergence guarantees.
\newblock {\em arXiv:1409.2623}.

\bibitem{dgw-spec}
F.~M\'emoli.
\newblock A spectral notion of {G}romov-{W}asserstein distances and related
  methods.
\newblock {\em Applied and Computational Mathematics.}, 30:363--401, 2011.

\bibitem{entropy}
S.~Mendelson.
\newblock A few notes on statistical learning theory.
\newblock In {\em Lecture Notes in Computer Science}, volume 2600, pages 1--40,
  2003.

\bibitem{OvsjanikovSG08}
M.~Ovsjanikov, J.~Sun, and L.~J. Guibas.
\newblock Global intrinsic symmetries of shapes.
\newblock {\em Comput. Graph. Forum}, 27(5):1341--1348, 2008.

\bibitem{rosenberg1997lrm}
S.~Rosenberg.
\newblock {\em {The Laplacian on a Riemannian Manifold: An Introduction to
  Analysis on Manifolds}}.
\newblock Cambridge University Press, 1997.

\bibitem{SS-iso}
Z.~Shi.
\newblock Point integral method for elliptic equation with isotropic
  coefficients with convergence guarantee.
\newblock {\em arXiv:1506.03606}.

\bibitem{SS-rate}
Z.~Shi and J.~Sun.
\newblock Convergence of laplacian spectra from point clouds.
\newblock {\em arXiv:1506.01788}.

\bibitem{SS-neumann}
Z.~Shi and J.~Sun.
\newblock Convergence of the point integral method for the poisson equation on
  manifolds i: the neumann boundary.
\newblock {\em arXiv:1403.2141}.

\bibitem{Singer13}
A.~Singer and H.~tieng Wu.
\newblock Spectral convergence of the connection laplacian from random samples.
\newblock {\em arXiv:1306.1587}.

\end{thebibliography}

\end{document}